\documentclass[showpacs,onecolumn,aps,notitlepage]{revtex4-1}
\usepackage{qcircuit}
\usepackage[dvips]{graphicx}
\usepackage{amsmath,amssymb,amsthm,mathrsfs,amsfonts,dsfont}
\usepackage{subfigure, epsfig}
\usepackage{braket}
\usepackage{bm}
\usepackage{enumerate}
\usepackage{algorithm}
\usepackage{algpseudocode}
\usepackage{diagbox}
\usepackage{physics}
\usepackage{color}
\usepackage[dvipsnames]{xcolor}
\usepackage{multirow}
\usepackage[marginal]{footmisc}
\usepackage{comment}
\usepackage{tikz}
\usepackage{cancel}
\usepackage[]{qcircuit}
\usepackage{soul}
\usepackage{pgfplots}
\usepackage{mathrsfs}
\usetikzlibrary{arrows}
\usepackage{makecell}
\usetikzlibrary{arrows}
\usetikzlibrary{shapes,fadings,snakes}
\usetikzlibrary{decorations.pathmorphing,patterns}
\usetikzlibrary{calc}
\usetikzlibrary{positioning}
\usepackage[colorlinks = true]{hyperref}

\graphicspath{{./figure/}}

\newtheorem{theorem}{Theorem}

\newtheorem{lemma}{Lemma}

\newtheorem{proposition}{Proposition}

\newcommand{\mc}{\mathcal}

\newcommand{\mbb}{\mathbb}

\newcommand{\comments}[1]{}

\newcommand{\E}{\mathcal{E}}

\begin{document}

\title{The sudden death of quantum advantage in correlation generations}
 \author{Weixiao Sun$^{1}$}\email{These authors contributed equally to this work.}
 \author{Fuchuan Wei$^{2,3}$}\email{These authors contributed equally to this work.}
 \author{Yuguo Shao$^{2,3}$}
 \author{Zhaohui Wei$^{2,4,}$}\email{weizhaohui@gmail.com}
 \affiliation{$^{1}$Institute for Interdisciplinary Information Sciences, Tsinghua University, Beijing 100084, China\\$^{2}$Yau Mathematical Sciences Center, Tsinghua University, Beijing 100084, China\\$^{3}$Department of Mathematics, Tsinghua University, Beijing 100084, China\\$^{4}$Yanqi Lake Beijing Institute of Mathematical Sciences and Applications, Beijing 101407, China}

\begin{abstract}
As quantum error corrections still cannot be realized physically, quantum noise is the most profound obstacle to the implementations of large-scale quantum algorithms or quantum schemes. It has been well-known that if a quantum computer suffers from too strong quantum noise, its running can be easily simulated by a classical computer, making the quantum advantage impossible. Generally speaking, however, the dynamical process that how quantum noise of varying strengths from 0 to a fatal level impacts and destroys quantum advantage has not been understood well. Undoubtedly, achieving this will be extremely valuable for us to understand the power of noisy intermediate-scale quantum computers. Meanwhile, correlation generation is a precious theoretical model of information processing tasks in which the quantum advantage can be precisely quantified. Here we show that this model also provides us a valuable insight into understanding the impact of noise on quantum advantage. Particularly, we will rigorously prove that when the strength of quantum noise continuously goes up from 0, the quantum advantage gradually declines, and eventually fades away completely. Surprisingly, in some cases we observe an unexpected phenomenon we call the sudden death of the quantum advantage, i.e., when the strength of quantum noise exceeds a certain point, the quantum advantage disappears suddenly from a non-negligible level. This interesting phenomenon reveals the tremendous harm of noise to quantum information processing tasks from a new viewpoint.
\end{abstract}

\date{\today}

\maketitle

\section{Introduction}

Quantum noise is the major obstacle in building large-scale quantum computers. Although remarkable progresses on their physical implementations have been made~\cite{arute2019quantum,wu2021strong,zhu2022quantum,zhong2020quantum,zhong2021phase}, we still do not have enough computational resources to perform quantum error corrections, which are believed to be the ultimate approach to fight against quantum noise. As a result, for now we can only work with the so-called noisy intermediate-scale quantum (NISQ) computers~\cite{preskill2018quantum,bharti2022noisy,chen2022complexity}, which are usually based on short-depth quantum circuits that do not involve any quantum error corrections. In recent years, in order to show that NISQ computers can efficiently solve computational problems that any classical computers cannot, researchers have been intensively utilizing them to sample the outputs of random quantum circuits and solve optimization problems like calculating the ground-state energy of physically relevant Hamiltonians. To reliably exhibit quantum advantage, one not only needs to improve the quality of quantum hardware and mitigate experimental imperfections such that NISQ computers have good precisions, but also needs to understand how precisely noise impacts the computational advantage of quantum computing. In this paper, we will focus on the second aspect.

To figure out the impact of noise on quantum advantage, we have to face two major challenges. First, the rigorous quantum advantage is usually hard to characterize mathematically. For example, the discovery of Shor's algorithm was a milestone in the developments of quantum computing, but in principle whether it can achieve exponential speedup over the optimal classical counterpart has not been rigorously proved or disproved~\cite{shor1994algorithms}, as the classical computational complexity of integer factorization has not been settled. In some specific computational models such as query complexity~\cite{grover1996fast,buhrman2002complexity}, communication complexity~\cite{yao1993quantum,brassard2003quantum}, and some machine learning models~\cite{gao2018quantum,huang2022quantum}, researchers have managed to identify the existence of quantum advantage mathematically, but this is often achieved by roughly upper bounding the power of quantum computing and meanwhile lower bounding that of its classical counterpart on a same task. Indeed, on most problems it is very hard to \emph{quantify} quantum advantage. To look into the impact of noise with different strengths on quantum advantage, a proper measure for quantifying it is desirable.

Second, the influences of noise on quantum circuits are highly complicated, and have not been understood well. In the past several years, because of the impact of the quantum supremacy experiments based on random quantum circuits~\cite{arute2019quantum,wu2021strong,zhu2022quantum}, a huge amount of efforts have been devoted to analyzing the properties of noisy random quantum circuits~\cite{aaronson2017complexity,aaronson2020classical,boixo2018characterizing,dalzell2022random,deshpande2022tight,aharonov2023polynomial,cheng2023efficient,hangleiter2023computational,napp2022efficient,mishra2024classically,aharonov2023polynomial,dalzell2024random,gao2024limitations,zhou2020limits,pan2022solving,yung2017can}. Since quantum advantage is certified only when a quantum computer can solve a computational problem much more easily than its classical counterparts, a crucial topic in these works is when sampling from the outputs of noisy random quantum circuits can and cannot be efficiently simulated by classical computers. For example, in Ref.\cite{aharonov2023polynomial} it was proved that in the presence of a constant rate of noise per quantum gate, sampling from the output distribution of a noisy random circuit can be approximately simulated by an efficient classical algorithm, as long as the so-called anti-concentration property is assumed~\cite{aaronson2013computational}. This surprising result highlights the importance of further decreasing error rates in NISQ computers. In Ref.\cite{dalzell2024random}, strong evidence was provided to show that under certain conditions, local errors are scrambled by random quantum circuits and can be approximated by global white noise, justifying an assumption that played a key role in Google's supremacy experiment~\cite{boixo2018characterizing,arute2019quantum}. Meanwhile, when demonstrating quantum supremacy based on random quantum circuits, a challenging issue is how to verify the existence of quantum advantage at a reasonable amount of experimental cost, for which several tools like the heavy-outcome generation fidelity~\cite{aaronson2017complexity} and the linear cross-entropy benchmark~\cite{arute2019quantum} have been proposed and studied~\cite{gao2024limitations,zhou2020limits,pan2022solving}.

In addition to sampling the outputs of random quantum circuits, very recently the quantum supremacy of NISQ computers working on optimization problems has also been investigated~\cite{stilck2021limitations,gonzalez2022error,de2023limitations}, and these works suggest that quantum advantage is possible only when the error rate is extremely low. Please see Refs.\cite{stilck2021limitations,gonzalez2022error,de2023limitations,gao2018efficient,noh2020efficient,ben2013quantum,fefferman2023effect,mele2024noise,wang2021noise,shao2023simulating} for more on how noise affects quantum computing.

Despite these encouraging progresses, however, in most cases the quantum advantage of noisy NISQ computers can only be identified in certain small or rough areas of the parameter spaces of quantum circuits, for example on one side of certain transition points~\cite{morvan2023phase,ware2023sharp,cheng2023efficient}. As a result, the dynamical process that how quantum noise with the strength continuously increasing from 0 to a fatal level harms the quantum advantage in a specific NISQ computer step by step, and eventually destroys it completely, is still far from being understood well. Undoubtedly, achieving this will be extremely important for us to develop NISQ computers with practical quantum supremacy.

In this paper, we will focus on a simpler sampling model than sampling from the outputs of noisy random quantum circuits, in which two separated parties, Alice and Bob, cooperate and try to sample two separated outcomes according to a target joint probability distribution~\cite{zhang2012game,Jain2013efficient}. This task can be fulfilled by either a quantum scheme or a classical one. In the former, Alice and Bob can realize the sampling by locally measuring a shared quantum seed state without any communications, and then outputting the corresponding outcomes; in the latter, Alice and Bob can replace the shared quantum state with a public randomness, and the other actions allowed are local classical operations only. A crucial feature of this elegant model is that a quantum scheme can enjoy a remarkable advantage over any classical one, furthermore, in the noiseless case the smallest computational costs needed for both scenarios, which are measured by the sizes of the shared quantum states and the public randomness respectively, can be characterized exactly~\cite{zhang2012game,Jain2013efficient}. In other words, the quantum advantage in this model can be precisely quantified. Based on this key fact, here we show that this model also allows us to rigorously analyze the dynamical process that how noise of continuously increasing strengths impacts the quantum advantage. Specifically, on an arbitrary target correlation, we first define a measure for the quantum advantage in this model, then we prove that when the strength of noise continuously goes up from zero, the quantum advantage gradually decreases, and eventually fades away completely. Surprisingly, we will report an unexpected phenomenon we call the \emph{sudden death of quantum advantage}, i.e., in some cases when the strength of noise reaches a certain point, the quantum advantage disappears suddenly from a non-negligible level. This interesting phenomenon brings us a  new viewpoint to appreciate the harm of noise in quantum information processing tasks.

\section{The Setting}

Suppose two separate parties, Alice and Bob, want to sample two outcomes $x\in[m]\equiv\{1,2,...,m\}$ and $y\in[n]$ respectively according to a discrete joint probability distribution $P=[P(x,y)]_{x\in[m],y\in[n]}$, i.e., the probability that Alice outputs $x$ and Bob outputs $y$ is $P(x,y)$. If this is the case, we say Alice and Bob generate $P$, and $P$ is also called a classical correlation. We can write $P$ as an $m\times n$ nonnegative matrix (A nonnegative matrix is a matrix with all the entries nonnegative), and without loss of generality in the current work we assume that $m=n$ for simplicity.
When generating $P$, we suppose that any communication between Alice and Bob is not allowed. Then if the target $P$ is not a product distribution, that is, $\operatorname{rank}(P)>1$, they need to share some initial resource to generate $P$. If Alice and Bob choose to generate $P$ with a classical protocol, the initial shared resource can be another shared classical correlation $P'$, which is called a \textit{seed} correlation. In this case, Alice and Bob can perform local classical operations on $P'$ to generate $P$. Similarly, they can also generate $P$ with a quantum protocol, in which the initial shared resource can be a shared quantum state $\rho$ called a \textit{seed} state. On this state, Alice and Bob perform proper local positive operator-valued measurements (POVM) on $\rho$ and output the corresponding outcomes to obtain $P$.
We define the size of $P'$ or $\rho$ to be half of the total number of bits or qubits needed to record $P'$ or $\rho$, and denote it as $\operatorname{size}(P')$ or $\operatorname{size}(\rho)$ respectively.
Please see Refs.\cite{zhang2012game,Jain2013efficient} for more details on this model of correlation generations.

For a given target $P$, we are concerned about the minimum value of $\operatorname{size}(P')$ (or $\operatorname{size}(\rho)$) for $P'$ (or $\rho$) that can generate $P$, which is called the \textit{classical (quantum) correlation complexity} of $P$, denoted $\mathrm{R}(P)$ (or $\mathrm{Q}(P)$), correspondingly.
It turns out that $\mathrm{R}(P)$ and $\mathrm{Q}(P)$ have been fully characterized~ \cite{zhang2012game, Jain2013efficient}. Specifically, it holds that
\begin{align*}
\mathrm{R}(P)=\left\lceil\log_2 \operatorname{rank}_{+}(P)\right\rceil
\end{align*}
and
\begin{align*}
\mathrm{Q}(P)&=\left\lceil\log_2 \operatorname{rank}_{\mathrm{psd}}(P)\right\rceil.
\end{align*}
For a nonnegative matrix $P\in\mbb{R}_{\ge0}^{n\times n}$, $\operatorname{rank}_+(P)$ is the nonnegative rank defined to be the minimum integer $r$ such that $P$ can be expressed as the sum of $r$ nonnegative rank-$1$ matrices.
$\operatorname{rank}_{\mathrm{psd}}(P)$ is the positive semidefinite rank (PSD-rank) defined to be the minimum integer $r$ such that there exist $r \times r$ positive semidefinite matrices $C_x, D_y \in \mathbb{C}^{r \times r}$ satisfying $P(x, y)=\tr\left(C_x D_y\right)$ for all $x$ and $y$, and this factorization is called a \textit{PSD decomposition} \cite{Samuel2012linear,Fawzi2015psdrank}. Although both of these two ranks are NP-hard to compute~\cite{vavasis2010complexity,shitov2017complexity}, this problem provides us with a precious model on which the quantum advantage can be quantified precisely. Indeed, it has been known that there exists such a classical correlation $P$ that $\operatorname{rank}_{+}(P)\gg\operatorname{rank}_{\mathrm{psd}}(P)$~\cite{Samuel2012linear}, implying that quantum schemes can have a remarkable advantage over classical ones in generating correlations.

In the quantum protocol, however, the states that Alice and Bob share may suffer from noise, due to various experimental imperfections, which may affect the above quantum advantage seriously. Analyzing this kind of affection is the major motivation of the current paper. For simplicity, here we suppose that all quantum noise is concentrated on the preparations and the quantum memories for initial seed states, and assume that all the other quantum operations are noiseless. Meanwhile, we assume that classical protocols are always noiseless. More specifically, in a quantum protocol we assume that the seed state $\rho$ suffers from a bipartite quantum noise of the form $\E_\lambda\otimes\E_\lambda$, i.e., each of the two subsystems of $\rho$ goes through a global depolarizing noise channel $\mc{E}_\lambda$ with a strength $\lambda$, which keeps the subsystems unchanged with probability $1-\lambda$ and replaces it with the maximally mixed state with probability $\lambda$. We would like to point out that this kind of noise can appear naturally in real-life quantum systems. For example, Ref.\cite{dalzell2024random} shows that under certain conditions the overall effect of local quantum errors in random quantum circuits can be approximated by global white noise. For simplicity, here we assume that all the other quantum operations are noiseless. When the strength $\lambda=0$, as mentioned a quantum scheme for correlation generations may exhibit a remarkable quantum advantage over its classical counterparts, while when $\lambda=1$, the initial seed state degrades into the white noise completely, making quantum advantage impossible (Alice and Bob can only produce a product distribution). In the rest of the paper, we will characterize the process that as the noise strength $\lambda$ continuously increases from $0$ to $1$, how the advantage of quantum schemes over classical ones gradually disappears.

If a classical correlation $P$ can be generated by $\rho$ under the noise channel $\E_\lambda\otimes\E_\lambda$, we denote this fact to be $\rho\xrightarrow{\lambda}P$.
Under such a noise channel, we define the \textit{quantum advantage} in generating a classical correlation P, denoted $\mathcal{S}_{\lambda}(P)$, as the ratio between the lowest classical cost of generating $P$ and the smallest size of $\rho$ such that $\rho\xrightarrow{\lambda}P$. Formally,
\begin{align*}
\mathcal{S}_{\lambda}(P)=\frac{\mathrm{R}(P)}{\inf_{\rho}\textup{size}(\rho)},
\end{align*}
where the infimum is taken over all $\rho$ such that $\rho\xrightarrow{\lambda}P$. It can be seen that this definition is consistent with our intuition on quantum advantage: If $\mathcal{S}_{\lambda}(P)$ is very large, say $10$, then the classical cost of generating $P$ is at least $9$ times higher than that of the optimal noisy quantum protocol that suffers from the noise channel $\E_\lambda\otimes\E_\lambda$.

Before proceeding, we first analyze the quantum advantage of noiseless quantum protocols in generating classical correlations. It has been known that for an arbitrary classical correlation $P$, the minimum quantum seed that can generate $P$ can always be chosen as a pure quantum state~\cite{sikora2016minimum}. Then we have the following conclusion, which shows that noiseless quantum protocols enjoy a remarkable quantum advantage in generating classical correlations.

\begin{theorem}\label{thm:PureEnt_advantage}
Suppose $\rho=\ketbra{\psi}{\psi}$ is a bipartite entangled pure quantum state on $\mc{H}_A\otimes\mc{H}_B$.
Then, there always exists a family of classical correlation $\{P_m\}_{m\in\mbb{Z}^+}$ such that $\rho\xrightarrow{0}P_m$ for all $m$, and  $\underset{m\rightarrow\infty}{\operatorname{lim}}\mathcal{S}_{0}(P_m)=\infty$.
\end{theorem}

The proof for Theorem~\ref{thm:PureEnt_advantage} can be seen in Appendix~\ref{app:proof_of_PureEnt}. Here we sketch the main idea for the proof. For $m$ distinct real numbers $\alpha_1,\cdots,\alpha_m$, we define a so-called $m\times m$ Euclidean distance matrix $\mathrm{EDM}_m$ by $\mathrm{EDM}_m(x,y)=\frac{(\alpha_x-\alpha_y)^2}{\sum_{ij}(\alpha_i-\alpha_j)^2}$. It has been known that $\operatorname{rank}_{\mathrm{psd}}(\mathrm{EDM}_m)=2$ and $\operatorname{rank}_{+}(\mathrm{EDM}_m)\ge2\sqrt{m}-2$~\cite{Shitov2019Euclidean}. Then we call $\mathrm{EDM}_m'\equiv L\cdot\mathrm{EDM}_m\cdot R$ a modified Euclidean distance matrix, where $L$ and $R$ are two $m\times m$ diagonal matrices with positive diagonal entries. We can prove that for any bipartite entangled pure state $\ket{\psi}$ and any integer $m$, $\ket{\psi}$ can always generate a class correlation $P_m$ such that some submatrix of $P_m$ is an $m\times m$ modified Euclidean distance matrix $\mathrm{EDM}_m'$. Then it can be seen
$\mathcal{S}_{0}(P_m)\ge\lceil\log_2(2\sqrt{m}-2)\rceil/\operatorname{size}(\rho)=\Omega(\log(m))$. Since $m$ can be chosen arbitrarily, the proof is completed.

\section{The Impact of Noise on Quantum Advantage}

We now move to the case of noisy quantum protocols that generate classical correlations, and try to find out how noise affects quantum advantage. Recall that the noise model we choose is that the global depolarizing noise channel $\E_\lambda(\rho)=(1-\lambda)\rho+\lambda \tr(\rho)I_d/d$ acts independently on the subsystems of Alice and Bob respectively, which transfers a noiseless seed state $\sigma$ to $\E_\lambda\otimes\E_\lambda(\sigma)$.

\subsection{The case of strong noise: the destruction of reachability}

Let us first consider the scenario of strong quantum noise. It turns out that for a classical correlation $P$, if the strength of quantum noise is higher than a certain value of $\lambda$, which we call the threshold, any quantum protocols will fail in generating $P$. Denote $\norm{\mathbf{s}}_1=\sum_{x=1}^n\abs{ s_x}$ as the one-norm of a vector $\mathbf{s}=( s_1,\cdots, s_n)$. For a Hilbert space $\mc{H}$, denote $D(\mc{H})$ be the set of density operators on $\mc{H}$. The following proposition can be utilized to identify this threshold. Note that if $\sigma\xrightarrow{\lambda}P$ for some $\lambda>0$, then every entry of $P$ is positive. Thus we always assume $P\in\mbb{R}_{>0}^{n\times n}$ below.
\begin{proposition}\label{reachable}
Suppose $P\in\mbb{R}^{n\times n}_{>0}$ is a correlation and $\lambda<1$ is a noise strength. Then there exists a dimension $d$ and a quantum state $\sigma\in D(\mbb{C}^{d}\otimes\mbb{C}^{d})$  making $\sigma\xrightarrow{\lambda}P$, if and only if there exists $\mathbf{s}=( s_1,\cdots, s_n),\mathbf{t}=( t_1,\cdots, t_n)\in\mbb{R}_{>0}^n$ with $\norm{\mathbf{s}}_1=\norm{\mathbf{t}}_1=1$ such that
\begin{align}\label{DefOfPLambmaHat}
    \hat{P}_\lambda^{\mathbf{s},\mathbf{t}}(x,y) \equiv P(x,y)-\lambda  s_x\sum_aP(a,y)-\lambda  t_y\sum_bP(x,b)+\lambda^2 s_x t_y\ge0\
\end{align}
holds for all $x,y$.
\end{proposition}

The proof for Proposition~\ref{reachable} can be found in Appendix~\ref{app:proof_reachable}.
We sketch the main idea of the proof here.
If there exist POVMs $\{E_x\}$ and $\{F_y\}$ such that
\begin{align*}
P(x,y)=&\tr(E_x\otimes F_y\E_\lambda\otimes\E_\lambda(\sigma))\\
=&(1-\lambda)^2\tr(E_x\otimes F_y\sigma)+\lambda(1-\lambda)\frac{\tr(E_x)}{d}\tr(F_y\sigma_B)\\
&+\lambda(1-\lambda)\tr(E_x\sigma_A)\frac{\tr(F_y)}{d}+\lambda^2\frac{\tr(E_x)}{d}\frac{\tr(F_y)}{d}.
\end{align*}
Take $ s_x=\tr(E_x)/d$ and $ t_y=\tr(F_y)/d$ and rearrange the terms, we can see that $\hat{P}_\lambda^{\mathbf{s},\mathbf{t}}(x,y)=(1-\lambda)^2\tr(E_x\otimes F_y\sigma)\ge0$.
On the other hand, if Eq.\eqref{DefOfPLambmaHat} holds, we can first find a noiseless quantum protocol with a seed state $\sigma'$, and POVMs $\{E_x'\}$ and $\{F_y'\}$ to generate the correlation $\frac{1}{(1-\lambda)^2}\hat{P}_\lambda^{\mathbf{s},\mathbf{t}}$. Then by enlarging the quantum systems of Alice and Bob to allow more freedom, we can construct a new quantum protocol that can generate $P$ exactly under the noise.

We now show how Proposition~\ref{reachable} allows us to characterize the region of $\lambda$ that a given classical correlation $P\in\mbb{R}_{>0}^{n\times n}$ is reachable for quantum protocols. For this, we denote $\Lambda(P)$ as this region, i.e.,
\begin{align}
\Lambda(P)=\left\{\lambda:\text{ there exists a quantum state } \rho \text{ such that }\rho\xrightarrow{\lambda}P\right\}.
\end{align}
For any $\lambda\in \Lambda(P)$, according to Proposition \ref{reachable} there exists $\mathbf{s},\mathbf{t}\in\mbb{R}_{>0}^n$ with $\norm{\mathbf{s}}_1=\norm{\mathbf{t}}_1=1$,  such that
$\hat{P}_\lambda^{\mathbf{s},\mathbf{t}}$ defined in Eq.\eqref{DefOfPLambmaHat} is a nonnegative matrix.
Note that
\begin{equation}\label{eq:derivative}
\begin{aligned}
\frac{\partial \hat{P}_\lambda^{\mathbf{s},\mathbf{t}}(x,y)}{\partial\lambda}&=2\lambda  s_x t_y- s_x\sum_aP(a,y)- t_y\sum_bP(x,b)\\
    &=-\frac{1}{1-\lambda}\left( s_x\sum_a \hat{P}_\lambda^{\mathbf{s},\mathbf{t}}(a,y)+ t_y\sum_b \hat{P}_\lambda^{\mathbf{s},\mathbf{t}}(x,b)\right)\le 0.
\end{aligned}
\end{equation}
Thus for any $0\le\lambda'\le\lambda$ we have $ \hat{P}_{\lambda'}^{\mathbf{s},\mathbf{t}}(x,y)\ge \hat{P}_\lambda^{\mathbf{s},\mathbf{t}}(x,y)\ge0$, which implies $\lambda'\in \Lambda(P)$. Therefore, $\Lambda(P)$ is an interval. That is, the continuous increase of $\lambda$ will destroy the reachability of quantum protocols for $P$, which will never come back.

Proposition~\ref{reachable} fully characterizes the reachability of noisy quantum protocols in generating correlations. However, for a given $P$, finding out whether suitable $\mathbf{s}$ and $\mathbf{t}$ can be found to satisfy the conditions in the Proposition~\ref{reachable} is very challenging, making it difficult to determine whether a correlated $P$ can be generated quantumly under the noise strength $\lambda$.
Because of this, below we provide an easy-to-handle necessary condition for the reachability of noisy quantum protocols. Specifically, for a given classical correlation $P\in\mbb{R}^{n\times n}_{\ge0}$, there exists a quantum state $\sigma\in D(\mbb{C}^{d}\otimes\mbb{C}^{d})$ satisfying $\sigma\xrightarrow{\lambda}P$, only if
\begin{align}\label{eq:reachabilitylemma}
    0\le\lambda\le1-\max_{\varphi\in S_n}\sum_{x=1}^n\sqrt{\max\left\{0,\sum_{b=1}^nP(x,b)\sum_{a=1}^nP(a,\varphi(x))-P(x,\varphi(x))\right\}},
\end{align}
where $S_n$ is a symmetric group of degree $n$. The proof can be found in Appendix~\ref{appendix_for_reachabilitylemma}.

Note that $\sum_{xy}\big(\sum_bP(x,b)\sum_{a}P(a,y)-P(x,y)\big)=0$. Then if $\operatorname{rank}(P)\ge2$, there always exist $x_0$ and $y_0$ such that
\begin{align*}
\sum_bP(x_0,b)\sum_{a}P(a,y_0)-P(x_0,y_0)>0.
\end{align*}
Thus the upper bound is always strictly less than $1$. That is to say, for any $P$ with $\operatorname{rank}(P)\ge2$, reachability will always be destroyed before $\lambda$ reaches $1$.

\subsection{The case of weak noise: the costs of noisy quantum protocols}
\label{subsection:cost_of_noisy_quantum_protocols}

For a given nontrivial classical correlation $P\in\mathbb{R}_{\ge0}^{n\times n}$ whose $\textup{rank}\geq2$,
we now focus on the interval $\Lambda(P)$, that is, the region of $\lambda$ that $P$ is reachable for noisy quantum protocols. In such a situation, the most important problem is how the cost of generating $P$ increases as the noise strength grows.

We denote $\mathrm{C}_{\lambda}(P)$ as the minimal local dimension for the seed states that can generate $P$ under the noise channel $\E_\lambda\otimes\E_\lambda$, i.e., the smallest number $d$ such that there exists $\sigma\in D(\mathbb{C}^d\otimes\mathbb{C}^d)$ satisfying $\sigma\xrightarrow{\lambda}P$. Then the size of the minimum seed state can be written as $\lceil\log_2\mathrm{C}_{\lambda}(P)\rceil$.

Recall that the cost of the optimal noiseless quantum protocol generating $P$ can be fully characterized as
\begin{align*}
\mathrm{Q}(P)&=\left\lceil\log_2 \operatorname{rank}_{\mathrm{psd}}(P)\right\rceil,
\end{align*}
where $\operatorname{rank}_{\mathrm{psd}}(P)$ is the PSD-rank of $P$~\cite{Jain2013efficient}. Inspired by the result above, we connect the cost of noisy quantum protocols in generating correlations to a modified version of  PSD-rank. Specifically, we define $\operatorname{rank}_{\textup{psd}}^{\E_\lambda}(P)$ to be the smallest $r$  for which there exist $r\times r$ PSD matrices $\{C_x\},\{D_y\}$ such that $P(x,y)=\tr(C_xD_y)$, satisfying that both
    \begin{align*}
        C_x-\frac{\lambda}{r}\tr(C_x\left(\sum_{k=1}^nC_k\right)^{-1})\sum_{k=1}^nC_k
    \end{align*}
    and
    \begin{align*}
        D_y-\frac{\lambda}{r}\tr(D_y\left(\sum_{k=1}^nD_k\right)^{-1})\sum_{k=1}^nD_k
    \end{align*}
    are PSD matrices. In this context, we say $\{C_i\},\{D_j\}$ is an $\E_\lambda$-PSD factorization of $P$.
\noindent Then we have that
\begin{align}\label{eq:cost_leq_ranknoisy}
\mathrm{C}_{\lambda}(P)\le\operatorname{rank}_{\textup{psd}}^{\E_\lambda}(P),
\end{align}
which will be proved in Appendix~\ref{app:new_psd_rank}. Particularly, if we restrict seed states to be pure states, the above inequality holds with equality.

Meanwhile, it turns out that Proposition~\ref{reachable} also helps to bound $\mathrm{C}_\lambda(P)$, which allows us to obtain the following upper and lower bounds for $\mathrm{C}_{\lambda}(P)$.
For a given classical correlation $P\in\mbb{R}_{>0}^{n\times n}$ and $\lambda\ge0$, if there exists $\mathbf{s},\mathbf{t}\in\mbb{R}_{>0}^n$ with $\norm{\mathbf{s}}_1=\norm{\mathbf{t}}_1=1$ such that $\hat{P}_\lambda^{\mathbf{s},\mathbf{t}}(x,y)$ defined in Eq.\eqref{DefOfPLambmaHat} is nonnegative for all $x,y$, then
\begin{align}\label{eq:upper_bound}
\mathrm{C}_{\lambda}(P)&\le\rank_\textup{psd}(\hat{P}_\lambda^{\mathbf{s},\mathbf{t}})\left\lceil\frac{1}{\min_{x,y}\{ s_x, t_y\}}\right\rceil,
\end{align}
and
\begin{align}\label{eq:lower_bound}
\mathrm{C}_{\lambda}(P)&\ge\frac{1}{1-\lambda}\left(\inf_{\mathbf{s}',\mathbf{t}'}\max_{x,y}\left\{\frac{\sum_bP(x,b)}{ s'_x},\frac{\sum_aP(a,y)}{ t'_y}\right\}-\lambda\right),
\end{align}
where $\mathbf{s}',\mathbf{t}'\in\mbb{R}_{>0}^n$ with $\norm{\mathbf{s}'}_1=\norm{\mathbf{t}'}_1=1$ satifying that $\hat{P}_\lambda^{\mathbf{s}',\mathbf{t}'}(x,y)\ge0$. The proofs for these bounds can be found in Appendix~\ref{app:proof_upper_lower_bound}.

\subsection{The decay process of quantum advantage}

Now we are ready to analyze how quantum advantage is affected when the strength of noise continuously grows up. We will demonstrate our techniques using specific examples.

Define $A_m\in\mathbb{R}^{(m+1)\times(m+1)}$ to be the correlation with the entries given by
    \begin{align}\label{P_example_decay}
        A_m = \left(\begin{array}{cccc}
        \frac{(1-q)^2}{2}&\frac{q(1-q)}{2m}&\cdots&\frac{q(1-q)}{2m}\\
        \frac{q(1-q)}{2m}&&&\\
        \vdots&&\frac{1+q^2}{2}B_m&\\
        \frac{q(1-q)}{2m}&&&
        \end{array}\right),
    \end{align}
    where $q=\frac{1}{1-k}-\sqrt{\frac{1}{(1-k)^2}-1}$, $k$ is a parameter, and $B_m$ is an $m\times m$ correlation with the entries being
    \begin{align*}
        B_m(x,y)=\frac{1}{m^2}\left(1-k\cos(2\pi\frac{x-1+y-1}{m})\right)
    \end{align*}
    for $x,y\in[m]$. Note that if $0<k<1$, we have $0<q<1$. In fact, the correlation $B_m$, as a matrix, is the slack matrix of two concentric regular polygons, as shown in Fig. \ref{fig:polygon}. More details about this geometric interpretation can be found in~\cite{fawzi2015positive,gillis2012geometric}, and some mathematical properties of $A_m$ and $B_m$ can be seen in Appendix~\ref{app:properties_of_B}.

\begin{figure}
\centering
\includegraphics[width=0.44\textwidth]{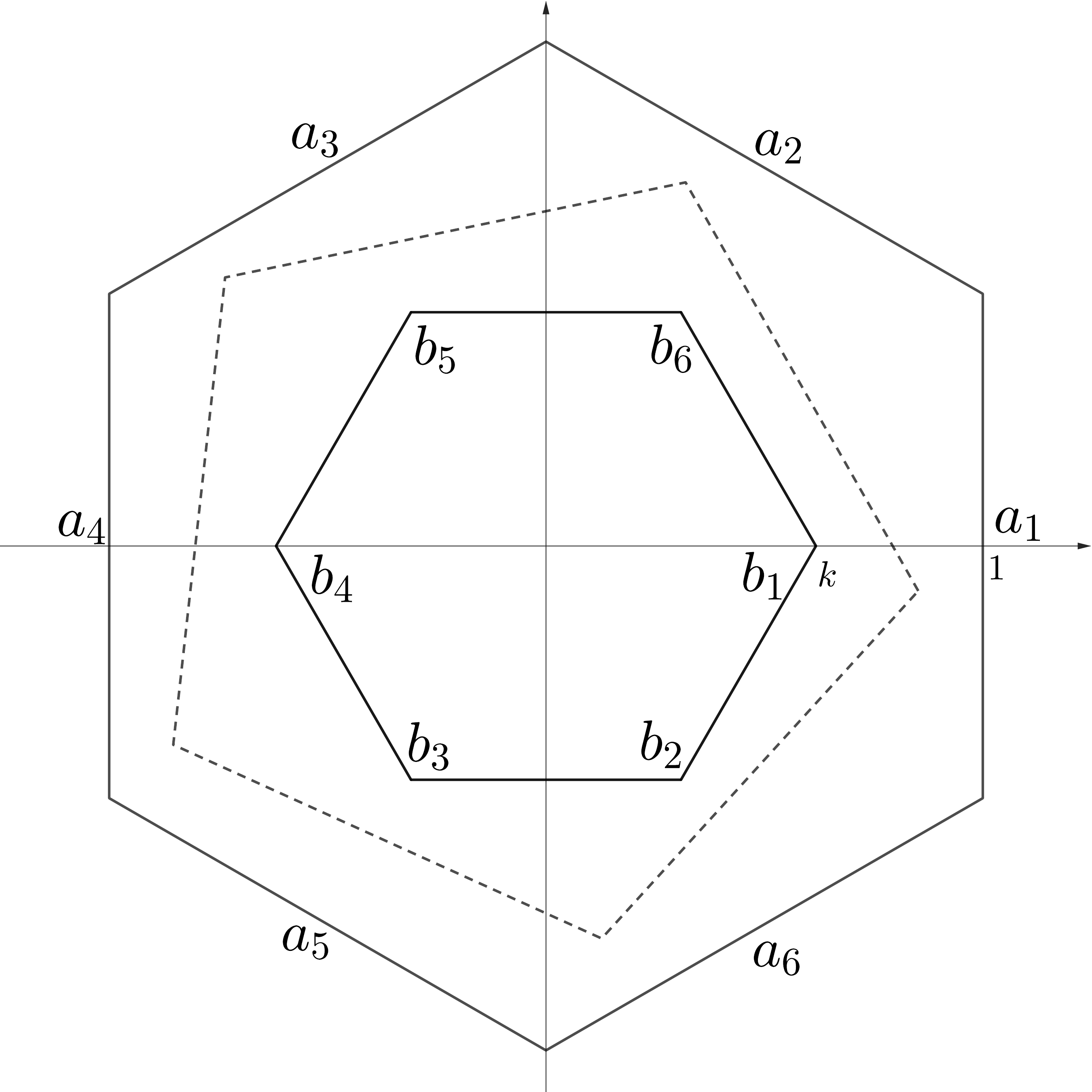}
\caption{A geometric interpretation for $B_m$, where $m=6$. The vertices of the smaller solid-line polygon are labeled clockwise from $1$ to $6$, while the edges of the bigger solid-line polygon are labeled counter-clockwise from $1$ to $6$. The dash-line pentagon shows a geometric interpretation for the restricted nonnegative rank of $B_m$.}
\label{fig:polygon}
\end{figure}

    It can be shown that in the absence of noise, $A_m$ has an arbitrarily large quantum advantage when $m\rightarrow\infty$.
    In fact, in Appendix~\ref{app:properties_of_B} we prove that if $k>\frac{\cos(2\pi/m)}{\cos^2(\pi/m)}$, then $\operatorname{rank}_+(A_m)> \log_2(m/2)$, whereas $\operatorname{rank}_{\textup{psd}}(A_m)\le 3$. As a result, it indicates that
    \begin{align}
        \underset{m\rightarrow\infty}{\lim}\mathcal{S}_{0}(A_m)=\underset{m\rightarrow\infty}{\lim}\frac{\mathrm{R}(A_m)}{\operatorname{rank}_{\textup{psd}}(A_m)}=\infty.
    \end{align}

We then turn to noisy quantum protocols that generate $A_m$. Recall that $\lceil\log_2\mathrm{C}_{\lambda}(A_m)\rceil$ is the smallest size of a seed state that can generate $A_m$ under the noise channel $\E_\lambda\otimes\E_\lambda$, which can be upper and lower bounded by Eq.\eqref{eq:upper_bound} and Eq.\eqref{eq:lower_bound}. As a result, we have the chance to characterize the decay process of quantum advantage in generating $A_m$ when $\lambda$ goes up.

\begin{theorem}\label{thm:A_decay_UB_LB}
As the noise strength $\lambda$ approaches $q$, the quantum advantage in generating $A_m$ decreases gradually to $0$.
\end{theorem}

\begin{proof}

    Let $\lambda=q-\epsilon$,  where $0<\epsilon<q$. And still use the notation
    \begin{align*}
        \hat{P}_\lambda^{\mathbf{s},\mathbf{t}}(x,y) = A_m(x,y)-\lambda  s_x\sum_{a=1}^{m+1}A_m(a,y)-\lambda  t_y\sum_{b=1}^{m+1}A_m(x,b)+\lambda^2 s_x t_y.
    \end{align*}
    According to the Proposition \ref{reachable}, there exists $\rho$ such that $\rho\xrightarrow{q-\epsilon}P$ if and only if there exist $\mathbf{s},\mathbf{t}\in\mbb{R}_{>0}^{m+1}$ with $\norm{\mathbf{s}}_1=\norm{\mathbf{t}}_1=1$ such that $\hat{P}_\lambda^{\mathbf{s},\mathbf{t}}(x,y)\ge0$.

    First, we prove that for any $0<\lambda<q$, $A_m$ can always be generated by some quantum state. Let
    \begin{align*}
    \mathbf{s}=\mathbf{t}=\big(\eta,(1-\eta)/m,(1-\eta)/m,\cdots,(1-\eta)/m\big),
    \end{align*}
    where $\eta=\min\left\{\frac{1-q}{2(q-\epsilon)},\frac{\epsilon(1+q)}{2q(q-\epsilon)}\right\}$. Note that $B_m(x,y)\ge\frac{1-k}{m^2}$. We have
    \begin{align*}
        \hat{P}_\lambda^{\mathbf{s},\mathbf{t}}(1,1)\ge& \frac{(1-q)^2}{2}-(q-\epsilon)\eta(1-q)\ge0,\\
        \hat{P}_\lambda^{\mathbf{s},\mathbf{t}}(1,y)=&\hat{P}_\lambda^{\mathbf{s},\mathbf{t}}(x,1)\ge\frac{q(1-q)}{2m}-(q-\epsilon)\eta\frac{1+q}{2m}-(q-\epsilon)\frac{1-\eta}{m}\frac{1-q}{2}\ge0,\ 2\le x,y\le m+1,\\
        \hat{P}_\lambda^{\mathbf{s},\mathbf{t}}(x,y)\ge&\frac{1+q^2}{2}\frac{1-k}{m^2}-2(q-\epsilon)\frac{1-\eta}{m}\frac{1+q}{2m}+(q-\epsilon)^2\frac{(1-\eta)^2}{m^2}\\
        \ge&\frac{1+q^2}{2}\frac{1-k}{m^2}-2q\frac{1}{m}\frac{1+q}{2m}+q^2\frac{1}{m^2}\\
        =&\frac{1}{m^2}\left(\frac{(1+q^2)(1-k)}{2}-q\right)=0,\ 2\le x,y\le m+1.
    \end{align*}
    That is, when $0<\lambda<q$, $\hat{P}_\lambda^{\mathbf{s},\mathbf{t}}(x,y)\ge0$ for any $x,y$, which means there exists a quantum state $\rho$ such that $\rho\xrightarrow{\lambda} A_m$.

    By Eq.\eqref{eq:upper_bound}, if
    $\epsilon$ is small enough, we can obtain
    \begin{align*}
        \mathrm{C}_{\lambda}(A_m)&\le\rank_\textup{psd}(\hat{P}_\lambda^{\mathbf{s},\mathbf{t}})\left\lceil\frac{1}{\min_{x,y}\{ s_x, t_y\}}\right\rceil
        \le m\left\lceil\frac{2q(q-\epsilon)}{\epsilon(1+q)}\right\rceil=\mc{O}(\epsilon^{-1}).
    \end{align*}

    Second, we now show that when $\epsilon\rightarrow0$, for any $\mathbf{s},\mathbf{t}\in\mbb{R}_{>0}^{m+1}$ satisfying $\norm{\mathbf{s}}_1=\norm{\mathbf{t}}_1=1$ and $\hat{P}_\lambda^{\mathbf{s},\mathbf{t}}(x,y)\ge0$ for all $x,y$, we must have $\min\{ s_1, t_1\}\rightarrow0$. In fact, it holds that
    \begin{align*}
        &\sum_{i=2}^{m+1}\hat{P}_\lambda^{\mathbf{s},\mathbf{t}}(i,1)+\sum_{j=2}^{m+1}\hat{P}_\lambda^{\mathbf{s},\mathbf{t}}(1,j)\\
        =&\sum_{i=2}^{m+1}A_m(i,1)-\sum_{i=2}^{m+1}\lambda  s_i\sum_{a=1}^{m+1}A_m(a,1)-\sum_{i=2}^{m+1}\lambda  t_1\sum_{b=1}^{m+1}A_m(i,b)+\sum_{i=2}^{m+1}\lambda^2 s_i t_1\\
        &+\sum_{j=2}^{m+1}A_m(1,j)-\sum_{j=2}^{m+1}\lambda  s_1\sum_{a=1}^{m+1}A_m(a,j)-\sum_{j=2}^{m+1}\lambda  t_j\sum_{b=1}^{m+1}A_m(1,b)+\sum_{j=2}^{m+1}\lambda^2 s_1 t_j\\
        =&\frac{q(1-q)}{2}-\lambda(1- s_1)\frac{1-q}{2}-\lambda  t_1\frac{1+q}{2}+\lambda^2(1- s_1) t_1+\frac{q(1-q)}{2}-\lambda  s_1\frac{1+q}{2}-\lambda(1- t_1)\frac{1-q}{2}+\lambda^2 s_1(1- t_1)\\
        \le&\frac{q(1-q)}{2}-\lambda(1- s_1)\frac{1-q}{2}-\lambda  t_1\frac{1+q}{2}+q^2(1- s_1) t_1+\frac{q(1-q)}{2}-\lambda  s_1\frac{1+q}{2}-\lambda(1- t_1)\frac{1-q}{2}+q^2 s_1(1- t_1)\\
        =&-2q^2 s_1 t_1+\epsilon-\epsilon(1- s_1- t_1)q\\
        \le&-2q^2 s_1 t_1+\epsilon(1+q).
    \end{align*}
    Therefore, for any $\mathbf{s},\mathbf{t}\in\mbb{R}_{>0}^{m+1}$ and $\norm{\mathbf{s}}_1=\norm{\mathbf{t}}_1=1$, in order to let $\hat{P}_\lambda^{\mathbf{s},\mathbf{t}}(x,y)\ge0$, we must have
    \begin{align*}
        \min\{ s_1, t_1\}\le\sqrt{\frac{\epsilon(1+q)}{2q^2}}.
    \end{align*}
    Then according to Eq.\eqref{eq:lower_bound},
    \begin{align*}
    \mathrm{C}_{\lambda}(A_m)&\ge\frac{1}{1-\lambda}\left(\inf_{\mathbf{s},\mathbf{t}}\max_{x,y}\left\{\frac{\sum_bA_m(x,b)}{s_x},\frac{\sum_aA_m(a,y)}{t_y}\right\}-\lambda\right)\\
    &\ge\frac{1}{1-\lambda}\left(\inf_{\mathbf{s},\mathbf{t}}\max\left\{\frac{\sum_bA_m(1,b)}{s_1},\frac{\sum_aA_m(a,1)}{t_1}\right\}-\lambda\right)\\
    &\ge\frac{1}{1-q+\epsilon}\left(\frac{1-q}{2}\sqrt{\frac{2q^2}{\epsilon(1+q)}}-q+\epsilon\right)=\Omega(\epsilon^{-1/2}).
    \end{align*}

    In conclusion, when $\epsilon$ is small, we have that
    \begin{align*}
        \mathcal{S}_{\lambda}(A_m)=\frac{\mathrm{R}(A_m)}{\lceil\log_2\mathrm{C}_{\lambda}(A_m)\rceil}=\Theta\left(\frac{1}{\log\epsilon^{-1}}\right).
    \end{align*}
    In other words, when $\lambda\rightarrow q$, $\mathrm{C}_{\lambda}(A_m)\rightarrow\infty$, implying that $\mathcal{S}_{\lambda}(A_m)\rightarrow 0$.
    This means that when $\lambda$ approaches $q$ from the left, although $A_m$ is always reachable for noisy quantum protocols, their costs will continuously increase without bound, implying that the noise with a strength around $q$ will completely destroy the remarkable quantum advantage of the noiseless quantum protocol generating $A_m$.
\end{proof}

\subsection{The sudden death of quantum advantage}

We have shown that $\mathcal{S}_{\lambda}(A_m)$ decays gradually to $0$ as the noise rate $\lambda$ grows. Surprisingly, if we instead focus on $B_m$, a subcorrelation of $A_m$, we will observe an unexpected phenomenon: as $\lambda$ continuously increases from $0$ and exceeds a certain point, $\mathcal{S}_{\lambda}(B_m)$ suddenly drops to $0$ from a non-negligible value. We call this kind of phenomenon the sudden death of quantum advantage in correlation generation.

\begin{theorem}\label{thm:B_sudden_death}
Let $0<k<1$. For any integer $m$ and the classical correlation $B_m$ defined in Eq.\eqref{P_example_decay}, we have
\begin{enumerate}
    \item $\mathcal{S}_{\lambda}(B_m)=\Omega\left(\log\log(m)\right)$, \text{ if } $0\le\lambda\le1-\sqrt{k}$;
    \item $\mathcal{S}_{\lambda}(B_m)=0$, \text{ if } $\lambda>1-\sqrt{k}$.
\end{enumerate}
Therefore, the quantum advantage in generating $B_m$ experiences sudden death when the noise strength exceeds $\lambda=1-\sqrt{k}$.
\end{theorem}

\begin{proof}
    Let $\varphi(1)=1,\varphi(x)=m+2-x$. We have $B_m(x,\varphi(x))=(1-k)/m^2$ for all $x$. Then
    \begin{align*}
        \sum_bB_m(x,b)\sum_{a}B_m(a,\varphi(x))-B_m(x,\varphi(x))=\frac{k}{m^2}.
    \end{align*}
    By Eq.\eqref{eq:reachabilitylemma}, if $B_m$ can be generated under noise $\E_\lambda\otimes\E_\lambda$, we must have
    \begin{align*}
        0\le\lambda\le1-\sum_{x=1}^m\sqrt{\max\left\{0,\sum_bB_m(x,b)\sum_{a}B_m(a,\varphi(x))-B_m(x,\varphi(x))\right\}}=1-\sqrt{k}.
    \end{align*}
    Thus for $\lambda>1-\sqrt{k}$, $B_m$ is not reachable for any noisy quantum protocols, meaning that $\mathcal{S}_{\lambda}(B_m)=\frac{\mathrm{R}(B_m)}{\infty}=0$.

    For $0\le\lambda\le 1-\sqrt{k}$, we selcet the $2\times 2$ hermitian matrices $\{C_x\}$ and $\{D_y\}$ introduced in Lemma~\ref{lemma:B_ranks} of Appendix~\ref{app:properties_of_B} as a PSD decomposition of $B_m$, then it holds that
    \begin{align*}
        &C_x-\frac{\lambda}{r}\tr(C_x\Big(\sum_{k=1}^nC_k\Big)^{-1})\sum_{k=1}^nC_k=C_x-\frac{\lambda}{\sqrt{2}m}I,\\
        &D_y-\frac{\lambda}{r}\tr(D_y\Big(\sum_{k=1}^nD_k\Big)^{-1})\sum_{k=1}^nD_k=D_y-\frac{\lambda}{\sqrt{2}m}I.
    \end{align*}
     Note that the minimal eigenvalues of $C_x$ and $D_y$ are all $\frac{1}{\sqrt{2}m}(1-\sqrt{k})$, which means both of $C_i-\frac{\lambda}{\sqrt{2}m}I$ and $D_j-\frac{\lambda}{\sqrt{2}m}I$ are PSD matrices, thus $\{C_x\}$ and $\{D_y\}$ form an $\E_\lambda$-PSD factorization of $B_m$. By Eq.\eqref{eq:cost_leq_ranknoisy}, we have
    \begin{align*}
    \mathrm{C}_{\lambda}(B_m)\le\operatorname{rank}_{\textup{psd}}^{\E_\lambda}(B_m)\le2.
    \end{align*}
    Since $B_m$ is not a product distribution, we have $\mathrm{C}_{\lambda}(B_m)=2$. That is, if $0\le\lambda\le 1-\sqrt{k}$, we have
    \begin{align*}
    \mathcal{S}_{\lambda}(B_m)=\frac{\mathrm{R}(B_m)}{\lceil\log_2\mathrm{C}_{\lambda}(B_m)\rceil}=\mathrm{R}(B_m),
    \end{align*}
    which is a constant independent of $\lambda$. By Lemma \ref{lemma:B_ranks} in Appendix~\ref{app:properties_of_B}, we have $\mathrm{R}(B_m)\ge\lceil\log_2\log_2(m/2)\rceil$, which completes the proof.
\end{proof}

For any classical correlation $P$, when the sudden death of quantum advantage happens to its generation, $\mathrm{C}_{\lambda}(P)$ will be upper bounded for any $\lambda\in \Lambda(P)$, which means that $\sup_{\lambda\in \Lambda(P)}\mathrm{C}_{\lambda}(P)<\infty$. To figure out when such a phenomenon occurs, we provide two equivalent conditions as below.

\begin{theorem}\label{thm:sudden_death}
    For any correlation $P\in\mbb{R}_{>0}^{n\times n}$, the following statements are equivalent to each other:
    \begin{enumerate}
        \item The quantum advantage in generating $P$ experiences the sudden death when $\lambda$ continuously goes up from $0$.
        \item $\Lambda(P)$ is a right closed interval.
        \item For any $\lambda>0$, if there exist non-negative vectors $\mathbf{s},\mathbf{t}\in\mbb{R}_{\ge0}^{n}$ with $\norm{\mathbf{s}}_1=\norm{\mathbf{t}}_1=1$ such that
        $\hat{P}_\lambda^{\mathbf{s},\mathbf{t}}$
        defined in Eq.\eqref{DefOfPLambmaHat} is a nonnegative matrix, then $P$ can be generated quantumly under the noise channel $\E_\lambda\otimes\E_\lambda$.
    \end{enumerate}
\end{theorem}

The proof for Theorem~\ref{thm:sudden_death} is deferred to Appendix~\ref{app:proof_of_SuddenSeath_equivalent}. As a simple application of this theorem, we now use it to reexamine the correlation $B_m$. According to the proof for Theorem~\ref{thm:B_sudden_death} we have that $\Lambda(B_m)=\big[0,1-\sqrt{k}\big]$, which is right-closed. Then Theorem~\ref{thm:sudden_death} immediately implies that the quantum advantage in generating $B_m$ suffers from the sudden death.

\section{Discussion}

Based on an elegant sampling model in which quantum advantage has been identified and even quantified, we mathematically characterize the dynamical process that noise of increasing strength gradually suppresses and even destroys the quantum advantage. We also report a surprising phenomenon that under the impact of noise, sudden death may happen to the quantum advantage in this model. To our knowledge, this is the first quantum protocol that sudden death is exhibited to quantum advantage, implying that in different quantum information processing tasks the impact of noise may have different behaviors. Our work suggests that, to fully understand the power of NISQ computers, we need to make further efforts to characterize the nature of quantum noise for different computational models of these computers, especially the quantum circuit model.

\begin{acknowledgments} We thank Xun Gao and Zhengfeng Ji for helpful discussions, and Zhengwei Liu for valuable comments. Weixiao Sun and Zhaohui Wei were supported in part by the National Natural Science
Foundation of China under Grant 62272259 and Grant 62332009; and
in part by Beijing Natural Science Foundation under Grant Z220002.
Fuchuan Wei and Yuguo Shao were supported by BMSTC and ACZSP (Grant No.~Z221100002722017).
\end{acknowledgments}

\bibliographystyle{naturemag}
\bibliography{ref}

\clearpage

\appendix

\setcounter{section}{0}
\setcounter{theorem}{0}
\setcounter{definition}{0}
\setcounter{proposition}{0}
\setcounter{corollary}{0}
\setcounter{lemma}{0}

\renewcommand{\thetheorem}{\Alph{section}.\arabic{theorem}}
\renewcommand{\thedefinition}{\Alph{section}.\arabic{definition}}
\renewcommand{\thelemma}{\Alph{section}.\arabic{lemma}}
\renewcommand{\theproposition}{\Alph{section}.\arabic{proposition}}

\section{The proof for Theorem~\ref{thm:PureEnt_advantage}}\label{app:proof_of_PureEnt}

\begin{theorem}
Suppose $\rho=\ketbra{\psi}{\psi}$ is a bipartite entangled pure quantum state on $\mc{H}_A\otimes\mc{H}_B$.
Then, there always exists a family of classical correlation $\{P_m\}_{m\in\mbb{Z}^+}$ such that $\rho\xrightarrow{0}P_m$ for all $m$, and  $\underset{m\rightarrow\infty}{\operatorname{lim}}\mathcal{S}_{0}(P_m)=\infty$.
\end{theorem}
\begin{proof}
Suppose $\dim(\mc{H}_A)=\dim(\mc{H}_B)=d\geq 2$,  the Schmidt decomposition of $\ket{\psi}$ is
\begin{align*}
\ket{\psi}=\sum_{i=1}^d\sqrt{\lambda_i}\ket{\phi_i^A}\otimes\ket{\phi_i^B},
\end{align*}
where $\sqrt{\lambda_1}\ge\cdots\ge\sqrt{\lambda_d}>0$ are the non-zero Schmidt coefficients.
Denote $\mu_1=\frac{\lambda_1}{\lambda_1+\lambda_2}$, $\mu_2=\frac{\lambda_2}{\lambda_1+\lambda_2}$, and let
\begin{align*}
\ket{\psi'}=\sqrt{\mu_1}\ket{\phi_1^A}\otimes\ket{\phi_1^B}+\sqrt{\mu_2}\ket{\phi_2^A}\otimes\ket{\phi_2^B}.
\end{align*}
Take $\Pi^A=\ketbra{\phi_1^A}{\phi_1^A}+\ketbra{\phi_2^A}{\phi_2^A}$, $\Pi^B=\ketbra{\phi_1^B}{\phi_1^B}+\ketbra{\phi_2^B}{\phi_2^B}$ be two projectors, we have $\Pi^A\otimes\Pi^B\ket{\psi}=\sqrt{\lambda_1+\lambda_2}\ket{\psi'}$.

For $m$ distinct real numbers $\alpha_1,\cdots,\alpha_m$, define the $m\times m$  Euclidean distance matrix $\mathrm{EDM}_m$ by $\mathrm{EDM}_m(x,y)=\frac{(\alpha_x-\alpha_y)^2}{\sum_{ij}(\alpha_i-\alpha_j)^2}$. For arbitrary $m$, we can pick $\alpha_1,\cdots,\alpha_m$ as described in \cite{Shitov2019Euclidean}, such that
\begin{enumerate}
    \item $\operatorname{rank}_{\mathrm{psd}}(\mathrm{EDM}_m)=2$ and $\operatorname{rank}_{+}(\mathrm{EDM}_m)\ge2\sqrt{m}-2$,
    \item $\alpha_1+\cdots+\alpha_m=0$,
    \item $\alpha_1$ big enough, such that $\sum_{y=1}^m\mathrm{EDM}_m(1,y)=\frac{\sum_y(\alpha_1-\alpha_y)^2}{\sum_{ij}(\alpha_i-\alpha_j)^2}>\mu_1-\frac{1}{2}$.
\end{enumerate}

For $x,y=1,\cdots,m$, let
\begin{align*}
C_x&=\frac{1}{\sqrt{2}}
\left(\begin{array}{cc}\frac{\alpha_x^2}{\sum_i\alpha_i^2} & \frac{-\alpha_x}{\sqrt{m\sum_i\alpha_i^2}} \\
\frac{-\alpha_x}{\sqrt{m\sum_i\alpha_i^2}} & \frac{1}{m}
\end{array}\right),\\
D_y&=\frac{1}{\sqrt{2}}
\left(\begin{array}{cc}\frac{1}{m} & \frac{\alpha_y}{\sqrt{m\sum_i\alpha_i^2}} \\
\frac{\alpha_y}{\sqrt{m\sum_i\alpha_i^2}} & \frac{\alpha_y^2}{\sum_i\alpha_i^2}
\end{array}\right).
\end{align*}
By $\sum_i\alpha_i=0$, we know that $\sum_{ij}(\alpha_i-\alpha_j)^2=2m\sum_i\alpha_i^2$, thus we have $\tr(C_xD_y)=\frac{(\alpha_x-\alpha_y)^2}{2m\sum_i\alpha_i^2}=\mathrm{EDM}_m(x,y)$. Meanwhile, it can be verified that
\begin{align*}
\sum_xC_x=\sum_yD_y=
\left(\begin{array}{cc}
\frac{1}{\sqrt{2}} &  \\
 & \frac{1}{\sqrt{2}}
\end{array}\right).
\end{align*}
Denote $r=\frac{\mu_1-1/2}{\sum_{y=1}^m\mathrm{EDM}_m(1,y)}$, then we have $0\le r<1$. Let $\mathrm{EDM}_m'=L\cdot\mathrm{EDM}_m\cdot R$, where
\begin{align*}
L&=\operatorname{diag}\left(1+r,1,\cdots,1\right),\\
R&=\operatorname{diag}\left(1-r,1,\cdots,1\right).
\end{align*}
Since $L$ and $R$ are diagonal matrices with diagonal elements $>0$, we have $\operatorname{rank}_+(\mathrm{EDM}_m')=\operatorname{rank}_+(\mathrm{EDM}_m)$. Note that
\begin{equation}\label{eq:PSD_decomposition_of_rescaled}
\begin{aligned}
&\{(1+r)C_1,C_2,\cdots,C_m\},\\
&\{(1-r)D_1,D_2,\cdots,D_m\}
\end{aligned}
\end{equation}
form a PSD decomposition of $\mathrm{EDM}_m'$. Since $C_1,D_1\in\mbb{C}^{2\times2}$ are PSD matrices satisfying $\tr(C_1D_1)=\mathrm{EDM}_m(1,1)=0$, we know that $C_1$ and $D_1$ are two rank-$1$ projectors with $C_1D_1=0$. Thus the eigenvalues of
\begin{align*}
\left((1+r)C_1+\sum_{x=2}^mC_x\right)\left((1-r)D_1+\sum_{y=2}^mD_y\right)=\left(rC_1+\frac{1}{\sqrt{2}}I_2\right)\left(-rD_1+\frac{1}{\sqrt{2}}I_2\right)=\frac{1}{2}I_2+\frac{r}{\sqrt{2}}(C_1-D_1)
\end{align*}
are
\begin{align*}
&\frac{1}{2}+\frac{r}{\sqrt{2}}\tr(C_1)=\frac{1}{2}+\frac{r}{\sum_{y=1}^m\mathrm{EDM}_m(1,y)}=\mu_1,\\
&\frac{1}{2}-\frac{r}{\sqrt{2}}\tr(D_1)=\frac{1}{2}-\frac{r}{\sum_{y=1}^m\mathrm{EDM}_m(1,y)}=\mu_2.
\end{align*}
Denote the diagonal form of PSD factorization of $\mathrm{EDM}_m'$ equivalent to Eq.\eqref{eq:PSD_decomposition_of_rescaled} by $\{C_x',D_y'\}$, where by equivalence we mean that there exists an invertible $H$ such that $C_x'=HC_xH^\dagger$ and $D_y'={(H^\dagger)}^{-1}D_yH^{-1}$ (such an $H$ always exists~\cite{lin2023all}). Then we have $\tr(C_x'D_y')=\mathrm{EDM}_m'(x,y)$ and
\begin{align*}
\sum_xC_x'=\sum_yD_y'=\left(\begin{array}{cc}
\sqrt{\mu_1} &  \\
 & \sqrt{\mu_2}
\end{array}\right),
\end{align*}
which coincides with the Schmidt coefficients of $\ket{\psi'}$. By Theorem 2 of \cite{chen2024generation}, there exist POVMs $\{E_x\}$ and $\{F_y\}$ on $\Pi^A\mc{H}_A\otimes\Pi^B\mc{H}_B$ (That is, $\Pi^AE_x\Pi^A=E_x$, $\Pi^BF_y\Pi^B=F_y$, $\sum_xE_x=\Pi^A$, and $\sum_yF_y=\Pi^B$) such that $\tr(E_x\otimes F_y\ketbra{\psi'}{\psi'})=\mathrm{EDM}_m'(x,y)$.

Now consider measuring POVMs
\begin{align*}
&\{E_1,\cdots,E_m,I-\Pi^A\},\\
&\{F_1,\cdots,F_m,I-\Pi^B\},
\end{align*}
on the two subsystems of $\ketbra{\psi}{\psi}$ respectively. Since
\begin{align*}
&\tr\left(E_x\otimes F_y\ketbra{\psi}{\psi}\right)=(\lambda_1+\lambda_2)\mathrm{EDM}_m'(x,y),\\
&\tr\left(E_x\otimes(I-\Pi^B)\ketbra{\psi}{\psi}\right)=0,\\
&\tr\left((I-\Pi^A)\otimes F_y\ketbra{\psi}{\psi}\right)=0,\\
&\tr\left((I-\Pi^A)\otimes(I-\Pi^B)\ketbra{\psi}{\psi}\right)=1-\lambda_1-\lambda_2,
\end{align*}
we have generated an $(m+1)\times(m+1)$ correlation
\begin{align*}
P_m = \left(\begin{array}{cc}
(\lambda_1+\lambda_2)\mathrm{EDM}_m' & 0  \\
0 & 1-\lambda_1-\lambda_2
\end{array}\right)
\end{align*}
with the seed state $\rho=\ketbra{\psi}{\psi}$. And it holds that
\begin{align*}
\operatorname{rank}_+(P_m)\ge\operatorname{rank}_+(\mathrm{EDM}_m')=\operatorname{rank}_+(\mathrm{EDM}_m)\ge2\sqrt{m}-2.
\end{align*}
Since $m$ can be arbitrary large, we have $\underset{m\rightarrow\infty}{\operatorname{lim}}\mathcal{S}_{0}(P_m)=\infty$.
\end{proof}

\section{The proof for Proposition~\ref{reachable}}
\label{app:proof_reachable}
\begin{proposition}
Suppose $P\in\mbb{R}^{n\times n}_{>0}$ is a correlation and $\lambda<1$ is a noise strength. Then there exists a dimension $d$ and a quantum state $\sigma\in D(\mbb{C}^{d}\otimes\mbb{C}^{d})$  making $\sigma\xrightarrow{\lambda}P$, if and only if there exists $\mathbf{s},\mathbf{t}\in\mbb{R}_{>0}^n$ with $\norm{\mathbf{s}}_1=\norm{\mathbf{t}}_1=1$ such that
\begin{align}\label{DefOfPLambmaHat_app}
    \hat{P}_\lambda^{\mathbf{s},\mathbf{t}}(x,y) \equiv P(x,y)-\lambda  s_x\sum_aP(a,y)-\lambda  t_y\sum_bP(x,b)+\lambda^2 s_x t_y\ge0\
\end{align}
holds for all $x,y$.
\end{proposition}
\begin{proof}
$(\Rightarrow)$ Note that
$$
\begin{aligned}
P(x,y)=&\tr(E_x\otimes F_y\E_{\lambda}\otimes\E_\lambda(\sigma))\\
=&(1-\lambda)^2\tr(E_x\otimes F_y\sigma)+\lambda(1-\lambda)\frac{\tr(E_x)}{d}\tr(F_y\sigma_B)\\
&+\lambda(1-\lambda)\tr(E_x\sigma_A)\frac{\tr(F_y)}{d}+\lambda^2\frac{\tr(E_x)}{d}\frac{\tr(F_y)}{d}\\
=&(1-\lambda)^2\tr(E_x\otimes F_y\sigma)+\lambda\frac{\tr(E_x)}{d}\left(\sum_aP(a,y)-\lambda \frac{\tr(F_y)}{d}\right)\\
&+\lambda\frac{\tr(F_y)}{d}\left(\sum_b P(x,b)-\lambda\frac{\tr(E_x)}{d}\right)+\lambda^2\frac{\tr(E_x)}{d}\frac{\tr(F_y)}{d}\\
=&(1-\lambda)^2\tr(E_x\otimes F_y\sigma)+\lambda\frac{\tr(E_x)}{d}\sum_aP(a,y)+\lambda\frac{\tr(F_y)}{d}\sum_b P(x,b)-\lambda^2\frac{\tr(E_x)}{d}\frac{\tr(F_y)}{d},
\end{aligned}
$$
where $\sigma_A=\tr_B(\sigma)$ and $\sigma_B=\tr_A(\sigma)$ are the reduced density matrices.
If we can take $ s_x=\frac{\tr(E_x)}{d}$, $ t_y=\frac{\tr(F_y)}{d}$, then
$$
\begin{aligned}
P(x,y)-\lambda  s_x\sum_aP(a,y)-\lambda  t_y\sum_bP(x,b)+\lambda^2 s_x t_y=(1-\lambda)^2\tr(E_x\otimes F_y\sigma)\ge0.
\end{aligned}
$$
With the properties of POVM, we have $\mathbf{s},\mathbf{t}\in\mbb{R}_{>0}^n$, and $\norm{\mathbf{s}}_1=\sum_x\frac{\tr(E_x)}{d}=\frac{\tr(I_d)}{d}=1$. Similarly, it holds that $\norm{\mathbf{t}}_1=1$.

$(\Leftarrow)$ Denote that
$$
    \hat{P}_\lambda^{\mathbf{s},\mathbf{t}}(x,y)=P(x,y)-\lambda  s_x\sum_aP(a,y)-\lambda  t_y\sum_bP(x,b)+\lambda^2 s_x t_y,
$$
which naturally satisfies $\sum_{xy}\frac{1}{(1-\lambda)^2}\hat{P}_\lambda^{\mathbf{s},\mathbf{t}}(x,y)=1$.

Hence, $\frac{1}{(1-\lambda)^2}\hat{P}_\lambda^{\mathbf{s},\mathbf{t}}$ can be regarded as a correlation. Based on the results in Ref.\cite{Jain2013efficient}, we know that there exist POVMs $\{E_x'\}$ and $\{F_y'\}$, and a quantum state $\sigma'\in D(\mbb{C}^{d'}\otimes\mbb{C}^{d'})$ such that $\tr(E_x'\otimes F_y'\sigma')=\frac{1}{(1-\lambda)^2}\hat{P}_\lambda^{\mathbf{s},\mathbf{t}}(x,y)$, where $d'=\operatorname{rank_{psd}}(\hat{P}_\lambda^{\mathbf{s},\mathbf{t}})$.

Take $k\in\mbb{Z}^+$ such that
$$
\begin{aligned}
d'k s_x-\tr(E_x')\ge0,\\
d'k t_y-\tr(F_y')\ge0,
\end{aligned}
$$
for all $x,y$. Denote $d=d'k$ and $\sigma=\ketbra{0}{0}\otimes\sigma'\otimes\ketbra{0}{0}$, where $\ketbra{0}{0}\in D(\mbb{C}^k)$. Let
$$
\begin{aligned}
E_x&=\ketbra{0}{0}\otimes E_x'+\left(\frac{d'k s_x-\tr(E_x')}{k-1}\right)( I_k-\ketbra{0}{0})\otimes \frac{ I_{d'}}{d'},\\
F_y&=F_y'\otimes\ketbra{0}{0}+\left(\frac{d'k t_y-\tr(F_y')}{k-1}\right)\frac{ I_{d'}}{d'}\otimes( I_k-\ketbra{0}{0}),
\end{aligned}
$$
which are PSD matrices. Note that
$$
\begin{aligned}
\sum_{x}E_x&=\ketbra{0}{0}\otimes\sum_x E_x'+\left(\frac{d'k\sum_x s_x-\tr(\sum_xE_x')}{k-1}\right)( I_k-\ketbra{0}{0})\otimes \frac{ I_{d'}}{d'}\\
&=\ketbra{0}{0}\otimes  I_{d'}+\left(\frac{d'k-d'}{k-1}\right)( I_k-\ketbra{0}{0})\otimes \frac{ I_{d'}}{d'}= I_k\otimes I_{d'}= I_d,\\
\sum_{y}F_y&=\sum_y F_y'\otimes\ketbra{0}{0}+\left(\frac{d'k\sum_bt_b-\tr(\sum_yF_y')}{k-1}\right)\frac{ I_{d'}}{d'}\otimes( I_k-\ketbra{0}{0})\\
&= I_{d'}\otimes\ketbra{0}{0}+\left(\frac{d'k-d'}{k-1}\right)\frac{ I_{d'}}{d'}\otimes( I_k-\ketbra{0}{0})= I_{d'}\otimes I_k= I_d,
\end{aligned}
$$
thus $\{E_x\}$, $\{F_y\}$ are valid POVMs.
We also have
$$
\begin{aligned}
\tr(E_x\otimes F_y\sigma)&=\tr(\ketbra{0}{0}\otimes E_x'\otimes F_y'\otimes\ketbra{0}{0}\cdot\ketbra{0}{0}\otimes\sigma'\otimes\ketbra{0}{0}+0+0+0)=\frac{1}{(1-\lambda)^2}\hat{P}_\lambda^{\mathbf{s},\mathbf{t}}(x,y),\\
\tr(E_x)&=\tr(E_x')+\left(\frac{d'k s_x-\tr(E_x')}{k-1}\right)(k-1)=d'k s_x=d s_x,\\
\tr(F_y)&=\tr(F_y')+\left(\frac{d'k t_y-\tr(F_y')}{k-1}\right)(k-1)=d'k t_y=d t_y.
\end{aligned}
$$

Finally, note that
\begin{align*}
\sum_a\hat{P}_\lambda^{\mathbf{s},\mathbf{t}}(a,y)&=(1-\lambda)\left(\sum_aP(a,y)-\lambda t_y\right),\\
\sum_b\hat{P}_\lambda^{\mathbf{s},\mathbf{t}}(x,b)&=(1-\lambda)\left(\sum_bP(x,b)-\lambda s_x\right).
\end{align*}

Thus we have
\begin{align*}
&\Tr(E_x\otimes F_y\E_{\lambda}\otimes\E_\lambda(\sigma))\\
=&(1-\lambda)^2\tr(E_x\otimes F_y\sigma)+\lambda(1-\lambda)\frac{\tr(E_x)}{d}\tr(F_y\sigma_B)
+\lambda(1-\lambda)\tr(E_x\sigma_A)\frac{\tr(F_y)}{d}+\lambda^2\frac{\tr(E_x)}{d}\frac{\tr(F_y)}{d}\\
=&\hat{P}_\lambda^{\mathbf{s},\mathbf{t}}(x,y)+\frac{\lambda}{1-\lambda} s_x\left(\sum_a\hat{P}_\lambda^{\mathbf{s},\mathbf{t}}(a,y)\right)+\frac{\lambda}{1-\lambda} t_y\left(\sum_b\hat{P}_\lambda^{\mathbf{s},\mathbf{t}}(x,b)\right)+\lambda^2 s_x t_y\\
=&\hat{P}_\lambda^{\mathbf{s},\mathbf{t}}(x,y)+\lambda s_x\left(\sum_aP(a,y)-\lambda t_y\right)+\lambda t_y\left(\sum_bP(x,b)-\lambda s_x\right)+\lambda^2 s_x t_y=P(x,y),
\end{align*}
which completes the proof.
\end{proof}

\section{The proof for Eq.\eqref{eq:reachabilitylemma}}\label{appendix_for_reachabilitylemma}

We restate the conclusion in Eq.\eqref{eq:reachabilitylemma}: For a given classical correlation $P\in\mbb{R}^{n\times n}_{\ge0}$, there exists a quantum state $\sigma\in D(\mbb{C}^{d}\otimes\mbb{C}^{d})$ satisfying $\sigma\xrightarrow{\lambda}P$,  only if
    \begin{align*}
        0\le\lambda\le1-\max_{\varphi\in S_n}\sum_{x=1}^n\sqrt{\max\left\{0,\sum_bP(x,b)\sum_{a}P(a,\varphi(x))-P(x,\varphi(x))\right\}},
    \end{align*}
where $S_n$ is a symmetric group of degree $n$.

\begin{proof}
By Proposition~\ref{reachable}, there exist $\mathbf{s},\mathbf{t}\in\mbb{R}_{>0}^n$ with $\norm{\mathbf{s}}_1=\norm{\mathbf{t}}_1=1$ satisfying that
\begin{equation}\label{eq:hat_p}
    \hat{P}_\lambda^{\mathbf{s},\mathbf{t}}(a,b)=P(a,b)-\lambda s_a\sum_aP(a,b)-\lambda t_b\sum_bP(a,b)+\lambda^2s_at_b\ge0.
\end{equation}
Then, summing $a$ and $b$ separately yields
\begin{align*}
    \sum_aP(a,y)&=
    \frac{1}{1-\lambda}\sum_a\hat{P}_\lambda^{\mathbf{s},\mathbf{t}}(a,y)+\lambda  t_y,\\
    \sum_bP(x,b)&=\frac{1}{1-\lambda}\sum_b\hat{P}_\lambda^{\mathbf{s},\mathbf{t}}(x,b)+\lambda  s_x.
\end{align*}
Thus by replacing $\sum_aP(a,y)$ and $\sum_bP(a,y)$ in Eq.\eqref{eq:hat_p}, $P(x,y)$ can be expressed as
\begin{align*}
P(x,y)&=\hat{P}_\lambda^{\mathbf{s},\mathbf{t}}(x,y)+\frac{\lambda}{1-\lambda} s_x\left(\sum_a\hat{P}_\lambda^{\mathbf{s},\mathbf{t}}(a,y)\right)+\frac{\lambda}{1-\lambda} t_y\left(\sum_b\hat{P}_\lambda^{\mathbf{s},\mathbf{t}}(x,b)\right)+\lambda^2 s_x t_y \\
&=\left(\hat{P}_\lambda^{\mathbf{s},\mathbf{t}}(x,y)-\frac{1}{(1-\lambda)^2}\sum_a\hat{P}_\lambda^{\mathbf{s},\mathbf{t}}(a,y)\sum_b\hat{P}_\lambda^{\mathbf{s},\mathbf{t}}(x,b)\right)+\left(\frac{1}{1-\lambda}\sum_a\hat{P}_\lambda^{\mathbf{s},\mathbf{t}}(a,y)+\lambda  t_y\right)\left(\frac{1}{1-\lambda}\sum_b\hat{P}_\lambda^{\mathbf{s},\mathbf{t}}(x,b)+\lambda  s_x\right)\\
&=\left(\hat{P}_\lambda^{\mathbf{s},\mathbf{t}}(x,y)-\frac{1}{(1-\lambda)^2}\sum_a\hat{P}_\lambda^{\mathbf{s},\mathbf{t}}(a,y)\sum_b\hat{P}_\lambda^{\mathbf{s},\mathbf{t}}(x,b)\right)+\sum_aP(a,y)\sum_bP(x,b).
\end{align*}
Then for all $x,y$, it holds that
\begin{align*}
    \frac{1}{(1-\lambda)^2}\sum_a\hat{P}_\lambda^{\mathbf{s},\mathbf{t}}(a,y)\sum_b\hat{P}_\lambda^{\mathbf{s},\mathbf{t}}(x,b)&=\hat{P}_\lambda^{\mathbf{s},\mathbf{t}}(a,b)-\left(P(x,y)-\sum_aP(a,y)\sum_bP(x,b)\right)\\
    &\ge-\left(P(x,y)-\sum_aP(a,y)\sum_bP(x,b)\right).
\end{align*}
Note that $\sum_a\hat{P}_\lambda^{\mathbf{s},\mathbf{t}}(a,y)\sum_b\hat{P}_\lambda^{\mathbf{s},\mathbf{t}}(x,b)\ge0$ and $\sum_{a,b}\hat{P}_\lambda^{\mathbf{s},\mathbf{t}}(a,b)=(1-\lambda)^2$. Then for any permutation $\varphi\in S_n$, we have
\begin{align*}
    2&\ge\frac{1}{(1-\lambda)^2}\sum_{x}\left(\sum_a\hat{P}_\lambda^{\mathbf{s},\mathbf{t}}(a,\varphi(x))+\sum_b\hat{P}_\lambda^{\mathbf{s},\mathbf{t}}(x,b)\right)\\
    &\ge\frac{2}{(1-\lambda)^2}\sum_{x}\sqrt{\sum_a\hat{P}_\lambda^{\mathbf{s},\mathbf{t}}(a,\varphi(x))\sum_b\hat{P}_\lambda^{\mathbf{s},\mathbf{t}}(x,b)}\\
    &\ge2\sum_{x}\sqrt{\max\left\{0,-\frac{1}{(1-\lambda)^2}\left(P(x,\varphi(x))-\sum_aP(a,\varphi(x))\sum_bP(x,b)\right)\right\}}\\
    &=\frac{2}{1-\lambda}\sum_{x}\sqrt{\max\left\{0,\sum_aP(a,\varphi(x))\sum_bP(x,b)-P(x,\varphi(x))\right\}}.
\end{align*}
This means that
\begin{align*}
    \lambda\le1-\sum_{x=1}^n\sqrt{\max\left\{0,\sum_bP(x,b)\sum_{a}P(a,\varphi(x))-P(x,\varphi(x))\right\}}.
\end{align*}
\end{proof}

\section{The proof for Eq.\eqref{eq:cost_leq_ranknoisy}}
\label{app:new_psd_rank}

We restate the conclusion in Eq.\eqref{eq:cost_leq_ranknoisy}: It holds that
\begin{align*}
\mathrm{C}_{\lambda}(P)\le\operatorname{rank}_{\textup{psd}}^{\E_\lambda}(P).
\end{align*}
\begin{proof}
To show why this is the case, suppose $r\times r$ PSD matrices $\{C_i\},\{D_j\}$ is an $\E_\lambda$-PSD factorization of $P$. Note that $\{C_i\},\{D_j\}$ is also a PSD factorization of $P$.
    We assume that $\{\widetilde{C_i}\},\{\widetilde{D_j}\}$ is a diagonal PSD factorization equivalent to $\{C_i\},\{D_j\}$. Recall that by equivalence we mean that there exists an invertible $H$ such that $\widetilde{C_i}=HC_iH^\dagger$ and $\widetilde{D_j}={(H^\dagger)}^{-1}D_jH^{-1}$. It turns out that such an $H$ always exists~\cite{lin2023all}.
    We can verify that $\{\widetilde{C_i}\},\{\widetilde{D_j}\}$ is also an $\E_\lambda$-PSD factorization for $P$, since $\{\widetilde{C_i}\},\{\widetilde{D_j}\}$ is also a PSD factorization of $P$, and
    \begin{align*}
        &\widetilde{C_i}-\frac{\lambda}{r}\tr(\widetilde{C_i}\left(\sum_{k=1}^n\widetilde{C_k}\right)^{-1})\sum_{k=1}^n\widetilde{C_k}\\
        =&H\left(C_i-\frac{\lambda}{r}\tr(HC_iH^\dagger\left(\sum_{k=1}^nHC_kH^\dagger\right)^{-1})\sum_{k=1}^n C_k\right)H^\dagger\\
        =&H\left(C_i-\frac{\lambda}{r}\tr(C_i\left(\sum_{k=1}^nC_k\right)^{-1})\sum_{k=1}^nC_k\right)H^\dagger
    \end{align*}
    is a PSD matrix (a similar conclusion holds for $\widetilde{D_j}$). Based on the above discussion, without loss of generality, we can suppose $\sum_k C_k=\sum_k D_k=\Lambda$, where $\Lambda=\textup{diag}(\beta_1,\dots,\beta_r)$ is an invertible diagonal matrix and $\tr(\Lambda^2)=1$.

    Let
    \begin{align*}
        M_i&=\frac{1}{1-\lambda}\sqrt{\Lambda^{-1}}\left(C_i^T-\frac{\lambda}{r}\tr(C_i\Lambda^{-1})\Lambda\right)\sqrt{\Lambda^{-1}},\\
        N_j&=\frac{1}{1-\lambda}\sqrt{\Lambda^{-1}}\left(D_j-\frac{\lambda}{r}\tr(D_j\Lambda^{-1})\Lambda\right)\sqrt{\Lambda^{-1}},\\
        |\psi_\Lambda\rangle&=\sqrt{\Lambda}\otimes\sqrt{\Lambda}\sum_{k=1}^r|kk\rangle=\sum_{k=1}^r\beta_k|kk\rangle.
    \end{align*}
    We have that $\{M_i\},\{N_j\}$ are valid POVMs, and $|\psi_\Lambda\rangle$ is a state. Furthermore, it holds that  $\tr(M_i)=\tr(C_i\Lambda^{-1})$ and $\tr(N_j)=\tr(D_j\Lambda^{-1})$.
    In addition, it can be verified that $\E_\lambda(M_i)=\sqrt{\Lambda^{-1}}C_i^T\sqrt{\Lambda^{-1}}$ and  $\E_\lambda(N_j)=\sqrt{\Lambda^{-1}}D_j\sqrt{\Lambda^{-1}}$.
    Then we have that
    \begin{align*}
        &\tr(\E_\lambda\otimes\E_\lambda(|\psi_\Lambda\rangle\langle\psi_\Lambda|)M_i\otimes N_j)\\
        =&\tr(|\psi_\Lambda\rangle\langle\psi_\Lambda|\E_\lambda(M_i)\otimes \E_\lambda(N_j))\\
        =&\tr(|\psi_\Lambda\rangle\langle\psi_\Lambda|\sqrt{\Lambda^{-1}}C_i^T\sqrt{\Lambda^{-1}}\otimes \sqrt{\Lambda^{-1}}D_j\sqrt{\Lambda^{-1}})\\
        =&\tr(\sum_{k=1}^r|kk\rangle\sum_{k=1}^r\langle kk|C_i^T\otimes D_j)\\
        =&\tr(C_i D_j)\\
        =&P_{ij}
    \end{align*}
    That is, $|\psi_\Lambda\rangle$ can generate $P$ under noise $\E_\lambda\otimes\E_\lambda$, which means the minimum local dimension of a seed state for $P$ under noise $\E_\lambda\otimes\E_\lambda$ is no more than $\operatorname{rank}_{\textup{psd}}^{\E_\lambda}(P)$.
\end{proof}

\section{The proofs for Eq.\eqref{eq:upper_bound} and Eq.\eqref{eq:lower_bound}}
\label{app:proof_upper_lower_bound}

We restate Eq.\eqref{eq:upper_bound} and Eq.\eqref{eq:lower_bound}:
Given $P\in\mbb{R}_{>0}^{n\times n}$ and $\lambda\ge0$, if there exists $\mathbf{s},\mathbf{t}\in\mbb{R}_{>0}^n$ with $\norm{\mathbf{s}}_1=\norm{\mathbf{t}}_1=1$ such that $\hat{P}_\lambda^{\mathbf{s},\mathbf{t}}(x,y)\ge0$ for all $x,y$, then
\begin{align*}
\mathrm{C}_{\lambda}(P)&\le\rank_\textup{psd}(\hat{P}_\lambda^{\mathbf{s},\mathbf{t}})\left\lceil\frac{1}{\min_{x,y}\{ s_x, t_y\}}\right\rceil,
\end{align*}
and
\begin{align*}
\mathrm{C}_{\lambda}(P)&\ge\frac{1}{1-\lambda}\left(\inf_{\mathbf{s}',\mathbf{t}'}\max_{x,y}\left\{\frac{\sum_bP(x,b)}{ s'_x},\frac{\sum_aP(a,y)}{ t'_y}\right\}-\lambda\right),
\end{align*}
where $\mathbf{s}',\mathbf{t}'\in\mbb{R}_{>0}^n$ with $\norm{\mathbf{s}'}_1=\norm{\mathbf{t}'}_1=1$ satisfy $\hat{P}_\lambda^{\mathbf{s}',\mathbf{t}'}(x,y)\ge0$.

\begin{proof}
(Proof of Eq.\eqref{eq:upper_bound})

As shown in the proof for Proposition~\ref{reachable}, if there exist $\mathbf{s},\mathbf{t}\in\mbb{R}_{>0}^n$ with $\norm{\mathbf{s}}_1=\norm{\mathbf{t}}_1=1$ such that $\hat{P}_\lambda^{\mathbf{s},\mathbf{t}}(x,y)\ge0$ for all $x,y$,
one can find POVMs $\{E_x'\}$ and $\{F_y'\}$, and a quantum state $\sigma'\in D(\mbb{C}^{d'}\otimes\mbb{C}^{d'})$ such that $\tr(E_x'\otimes F_y'\sigma')=\frac{1}{(1-\lambda)^2}\hat{P}_\lambda^{\mathbf{s},\mathbf{t}}(x,y)$, where $d'=\operatorname{rank_{psd}}(\hat{P}_\lambda^{\mathbf{s},\mathbf{t}})$.

By taking $k\in\mbb{Z}^+$ such that
\begin{equation}\label{eq:take_k}
\begin{aligned}
d'k s_x-\tr(E_x')\ge0,\\
d'k t_y-\tr(F_y')\ge0,
\end{aligned}
\end{equation}
for all $x,y$, we can construct a $d\times d$ seed state $\sigma=\ketbra{0}{0}\otimes\sigma'\otimes\ketbra{0}{0}$, $\ketbra{0}{0}\in D(\mbb{C}^k)$, and POVMs
$$
\begin{aligned}
E_x&=\ketbra{0}{0}\otimes E_x'+\left(\frac{d'k s_x-\tr(E_x')}{k-1}\right)( I_k-\ketbra{0}{0})\otimes \frac{ I_{d'}}{d'},\\
F_y&=F_y'\otimes\ketbra{0}{0}+\left(\frac{d'k t_y-\tr(F_y')}{k-1}\right)\frac{ I_{d'}}{d'}\otimes( I_k-\ketbra{0}{0}),
\end{aligned}
$$
to generate $P$, i.e., $P(x,y)=
\tr(E_x\otimes F_y\E_{\lambda}\otimes\E_\lambda(\sigma))$, where $d=d'k$.

In the above process, the dimension for the seed state is $d=d'k=\rank_\textup{psd}(\hat{P}_\lambda^{\mathbf{s},\mathbf{t}})k$, thus we have $\mathrm{C}_{\lambda}(P)\le\rank_\textup{psd}(\hat{P}_\lambda^{\mathbf{s},\mathbf{t}})k$. Actually, to ensure that Eq.\eqref{eq:take_k} holds, it suffices to take
\begin{align*}
k=\Bigg\lceil\underset{x,y}{\max}\left\{\frac{\tr(E_x')}{\rank_\textup{psd}(\hat{P}_\lambda^{\mathbf{s},\mathbf{t}})}\frac{1}{ s_x},\frac{\tr(F_y')}{\rank_\textup{psd}(\hat{P}_\lambda^{\mathbf{s},\mathbf{t}})}\frac{1}{ t_y}\right\}\Bigg\rceil\le\left\lceil\frac{1}{\min_{x,y}\{ s_x, t_y\}}\right\rceil.
\end{align*}

(Proof of Eq.\eqref{eq:lower_bound})

Suppose for a local dimension $d$, a quantum state $\sigma\in D(\mbb{C}^{d}\otimes\mbb{C}^{d})$, and POVMs $\{E_x,F_y\}$, it holds that $P(x,y)=\tr(E_x\otimes F_y\E_\lambda\otimes\E_\lambda(\sigma))$. Let $ s_x=\tr(E_x)/d$ and $ t_y=\tr(F_y)/d$. By Proposition~\ref{reachable} we have $\hat{P}_\lambda^{\mathbf{s},\mathbf{t}}(x,y)$ defined in Eq.\eqref{DefOfPLambmaHat} is nonnegative for all $x,y$.
Denote $\sigma_A=\tr_B(\sigma)$ and $\sigma_B=\tr_A(\sigma)$. For arbitrary $x$, we have
\begin{align*}
\sum_bP(x,b)=\tr(E_x\E_\lambda(\sigma_A))=(1-\lambda)\tr(E_x\sigma_A)+\lambda\frac{\tr(E_x)}{d}\le(1-\lambda)\tr(E_x)+\lambda\frac{\tr(E_x)}{d}=(d(1-\lambda)+\lambda) s_x,
\end{align*}
thus $ s_x\ge\frac{\sum_bP(x,b)}{d(1-\lambda)+\lambda}$, which leads to $d\geq \frac{1}{1-\lambda}\left(\frac{\sum_bP(x,b)}{ s_x}-\lambda\right)$. Similarly, for arbitrary $y$, we can obtain $d\ge\frac{1}{1-\lambda}\left(\frac{\sum_aP(a,y)}{ t_y}-\lambda\right)$, which leads to
\begin{align}\label{eq:proof_of_C_LB}
d&\ge\frac{1}{1-\lambda}\left(\max_{x,y}\left\{\frac{\sum_bP(x,b)}{ s_x},\frac{\sum_aP(a,y)}{ t_y}\right\}-\lambda\right).
\end{align}
On the both sides of Eq.\eqref{eq:proof_of_C_LB}, taking infimum over all possible quantum protocols which can generate $P$, we get
\begin{align*}
\mathrm{C}_{\lambda}(P)\ge\frac{1}{1-\lambda}\left(\inf_{\mathbf{s}',\mathbf{t}'}\max_{x,y}\left\{\frac{\sum_bP(x,b)}{ s'_x},\frac{\sum_aP(a,y)}{ t'_y}\right\}-\lambda\right).
\end{align*}

\end{proof}

\section{Some properties of $A_m$ and $B_m$}\label{app:properties_of_B}

Recall that $A_m\in\mathbb{R}^{(m+1)\times(m+1)}$ is a classical correlation with the entries given by
\begin{align*}
    A_m = \left(\begin{array}{cccc}
    \frac{(1-q)^2}{2}&\frac{q(1-q)}{2m}&\cdots&\frac{q(1-q)}{2m}\\
    \frac{q(1-q)}{2m}&&&\\
    \vdots&&\frac{1+q^2}{2}B_m&\\
    \frac{q(1-q)}{2m}&&&
    \end{array}\right),
\end{align*}
where $k\in (0,1)$ is a variable, $q=\frac{1}{1-k}-\sqrt{\frac{1}{(1-k)^2}-1}$, and $B_m\in\mathbb{R}^{m\times m}$ is also a classical correlation given by
\begin{align*}
    B_m(x,y)=\frac{1}{m^2}\left(1-k\cos(2\pi\frac{x-1+y-1}{m})\right),x,y=1,2,\cdots,m.
\end{align*}

As shown in Fig.~\ref{fig:polygon} of the main text, $B_m $ represents the slack matrix associated with two concentric regular polygons, denoted by $ R_{\mathrm{out}} $ and $ R_{\mathrm{in}} $ respectively.
Here $ R_{\mathrm{out}} $ has an inscribed circle with a radius of 1, while $ R_{\mathrm{in}} $ has a circumscribed circle with a radius $ k < 1 $.
The vertices of $ R_{\mathrm{in}} $ are positioned at the midpoints of the sides of $R_{\mathrm{out}}$. More details about the slack matrix can be found in \cite{fawzi2015positive}.

The outer polygon can be expressed as $R_{\mathrm{out}}=\{v\in\mbb{R}^2 | a_x^T v\leq 1 ,x=1,\cdots,m\}$, where
\begin{align}\label{eq:Rout}
a_x=\left(\cos{\frac{2(x-1)\pi}{m}},\sin{\frac{2(x-1)\pi}{m}}\right)
\end{align}
is the midpoint of $x$-th edge of the large polygon. The $y$-th vertices of the smaller polygon $R_{\mathrm{in}}$ is
\begin{align}\label{eq:Rin}
b_y=\left(k\cos{\frac{2(y-1)\pi}{m}},-k\sin{\frac{2(y-1)\pi}{m}}\right).
\end{align}
Thus the slack matrix between $R_{\mathrm{out}}$ and $R_{\mathrm{in}}$ is $S_{B,A}=(1-a_xb_y^T)_{xy}=\left(1-k\cos(2\pi\frac{x-1+y-1}{m})\right)_{xy=1,\cdots,m}$, which leads to that $B_m=\frac{1}{m^2}S_{B,A}$.

\begin{lemma}\label{lemma:B_ranks}
For $B_m$ defined above, we have
\begin{enumerate}
    \item $\operatorname{rank}_+(B_m)> \log_2(m/2)$, when $k>\frac{\cos(2\pi/m)}{\cos^2(\pi/m)}$,
    \item $\operatorname{rank}_{\textup{psd}}(B_m)=2$.
\end{enumerate}
\end{lemma}
\begin{proof}
(Part 1)

By the geometric interpretation of $B_m$ we can derive a lower bound for $\operatorname{rank}_+(B_m)$. We first show that $\textup{rank}_+(B_m)> \log_2 l$ if $k>\frac{\cos(\pi/l)}{\cos^2(\pi/m)}$.

A lower bound for $\textup{rank}_+(B_m)$ is derived from $\textup{rank}_+^*(B_m)$, known as the restricted nonnegative rank. This quantity represents the minimum number $k$ of vertices of a convex polygon $T \in \mathbb{R}^2$ such that $R_{\mathrm{in}} \subset T \subset R_{\mathrm{out}}$, where $R_{\mathrm{out}}$ and $R_{\mathrm{in}}$ denote the polygons described in Eq.\eqref{eq:Rout} and \eqref{eq:Rin} respectively. More details can be found in \cite{gillis2012geometric}.

A necessary condition for an $l$-gon $T$ to be put between $R_{\mathrm{in}}$ and $R_{\mathrm{out}}$ is $ R_{\mathrm{in}}^C \subset T \subset R_{\mathrm{out}}^C $, where $ R_{\mathrm{out}}^C $ denotes the circumscribed circle of $R_{\mathrm{out}}$ with a radius of $\frac{1}{\cos{\frac{\pi}{m}}}$, and $R_{\mathrm{in}}^C$ denotes the inscribed circle of $R_{\mathrm{in}}$ with a radius of $k\cos{\frac{\pi}{m}}$. Therefore, every angle of such a $T$ is at least $2\arcsin(k\cos^2\frac{\pi}{m})$, while the smallest angle of any $l$-gon is at most $\pi-\frac{2\pi}{l}$. Consequently, it must hold that $k\le\frac{\cos(\pi/l)}{\cos^2(\pi/m)}$. In other words, if $ k > \frac{\cos(\pi/l)}{\cos^2(\pi/m)} $, no $ l $-gon can be put between $R_{\mathrm{in}}$ and $R_{\mathrm{out}}$. Thus, $\textup{rank}_+^*(B_m)>l$.
According to Theorem 6 in \cite{gillis2012geometric}, we have that $\rank_+(B_m)\geq \log_2(\rank^*_+(B_m))> \log_2(l)$.

By taking $l=m/2$, it follows that $\rank_+(B_m)> \log_2(l)=\log_2(m/2)$ when $k>\frac{\cos(\pi/l)}{\cos^2(\pi/m)}=\frac{\cos(2\pi/m)}{\cos^2(\pi/m)}$, where $\frac{\cos(2\pi/m)}{\cos^2(\pi/m)}<1$.
Thus for any $m\ge3$, we can always choose $k$ close enough to $1$ such that $\rank_+(B_m)>\log_2(m/2)$.

(Part 2)

We can take
\begin{equation}\label{eq:psd_of_B}
\begin{aligned}
&C_x=\frac{1}{\sqrt{2}m}\left(\begin{array}{cc}
1 & \sqrt{k}e^{\frac{x-1}{m}2\pi i} \\
\sqrt{k}e^{-\frac{x-1}{m}2\pi i} & 1
\end{array}\right),\\
&D_y=\frac{1}{\sqrt{2}m}\left(\begin{array}{cc}
1 & -\sqrt{k}e^{-\frac{y-1}{m}2\pi i} \\
-\sqrt{k}e^{\frac{y-1}{m}2\pi i} & 1
\end{array}\right),
\end{aligned}
\end{equation}
where $i=\sqrt{-1}$.
It is straightforward to check $B_m(x,y)=\tr(C_xD_y)$ for all $x,y=1,\cdots,m$, which means that $\operatorname{rank}_{\textup{psd}}(B_m)\le2$.
Also note that $\rank(B_m)\ge2$, thus we have that $\operatorname{rank}_{\textup{psd}}(B_m)=2$.
\end{proof}

Based on the properties of $B_m$, we can now characterize $A_m$ accordingly.

\begin{lemma}
For $A_m$ definced in Eq.\eqref{P_example_decay}, we have
\begin{enumerate}
    \item $\operatorname{rank}_+(A_m)> \log_2(m/2)$, when $k>\frac{\cos(2\pi/m)}{\cos^2(\pi/m)}$,
    \item $\operatorname{rank}_{\textup{psd}}(A_m)\le 3$.
\end{enumerate}
\end{lemma}

\begin{proof}
(Part 1)

Since $B_m$ is a submatrix of $A_m$, we have $\operatorname{rank}_+(A_m)\ge\operatorname{rank}_+(B_m)$.
Meanwhile, according to Lemma~\ref{lemma:B_ranks}, it holds that $\operatorname{rank}_+(B_m)> \log_2(m/2)$, if $k>\frac{\cos(2\pi/m)}{\cos^2(\pi/m)}$.

(Part 2)

We take the same $\{C_x,D_y\}$ as in Eq.\eqref{eq:psd_of_B}, then let
    \begin{align*}
        C_1'&=D_1'=
        \frac{1-q}{2\sqrt{1+q^2}}
        \left(\begin{array}{ccc}
             \sqrt{2}&&  \\
             &q&\\
             &&q
        \end{array}\right),\\
        C_x'&=\sqrt{\frac{1+q^2}{2}}\left(\begin{array}{cc}
             0&  \\
             &C_{x-1}
        \end{array}\right) \text{ for } 2\le x\le m+1,\\
        D_y'&=\sqrt{\frac{1+q^2}{2}}\left(\begin{array}{cc}
             0&  \\
             &D_{y-1}
        \end{array}\right) \text{ for } 2\le y\le m+1.
    \end{align*}
We have
    \begin{align*}
        \tr(C_1'D_1')&=\frac{(1-q)^2}{2},\\
        \tr(C_x'D_1')&=\tr(C_1'D_y')=\frac{q(1-q)}{2m} \text{ for } 2\le x,y\le m+1,\\
        \tr(C_x'D_y')&=\frac{1+q^2}{2}\tr(C_{x-1}D_{y-1})=\frac{1+q^2}{2}(P_{m,k})_{x-1,y-1} \text{ for } 2\le x,y\le m+1.
    \end{align*}
    Thus $C_x',D_y'$ is a PSD decomposition of $A_m$, implying that $\operatorname{rank}_{\textup{psd}}(A_m)\le 3$.
\end{proof}

\section{The mathematical characterization for sudden death}\label{app:proof_of_SuddenSeath_equivalent}

We prove the following result, which is equivalent to Theorem~\ref{thm:sudden_death}.

\begin{theorem}\label{thm:sudden_death_app}
    For any correlation $P\in\mbb{R}_{>0}^{n\times n}$, the following statements are equivalent:
    \begin{enumerate}
        \item $\sup_{\lambda\in \Lambda(P)}\mathrm{C}_{\lambda}(P)=\infty$.
        \item $\Lambda(P)$ is a right open interval.
        \item There exists $\lambda>0$, such that there exist $\mathbf{s},\mathbf{t}\in\mbb{R}_{\ge0}^{n}$ with $\norm{\mathbf{s}}_1=\norm{\mathbf{t}}_1=1$ such that $\hat{P}_\lambda^{\mathbf{s},\mathbf{t}}(x,y)$ defined in Eq.\eqref{DefOfPLambmaHat} is nonnegative for all $x,y$. However, $P$ is not reachable for any quantum protocols under noise $\E_\lambda\otimes\E_\lambda$.
    \end{enumerate}
\end{theorem}

\begin{proof}

``$1\Rightarrow2$": If $\Lambda(P)$ is right closed, we choose $\lambda_u=\sup \Lambda(P)\in \Lambda(P)$. Thus there exist $\mathbf{s},\mathbf{t}\in\mbb{R}_{>0}^{n}$ with $\norm{\mathbf{s}}_1=\norm{\mathbf{t}}_1=1$ such that $\hat{P}_{\lambda_u}^{\mathbf{s},\mathbf{t}}(x,y)\ge0$. According to Eq.\eqref{eq:upper_bound}, $\mathrm{C}_{{\lambda_u}}(P)$ has an upper bound.

Since for any $\lambda\le\lambda_u$, it holds that $\mathrm{C}_{{\lambda}}(P)\le\mathrm{C}_{{\lambda_u}}(P)$, i.e., Alice and Bob can simulate strong noise with weak noise by introducing more noise. Then we have that $\sup_{\lambda\in \Lambda(P)}\mathrm{C}_{{\lambda}}(P)$ has an upper bound, which contradicts with $\sup_{\lambda\in \Lambda(P)}\mathrm{C}_{\lambda}(P)=\infty$.

``$2\Rightarrow3$": Let
\begin{align*}
    ST_\lambda=\left\{(\mathbf{s},\mathbf{t})\in\mathbb{R}_{\ge0}^n\times\mathbb{R}_{\ge0}^n\right|\hat{P}_\lambda^{\mathbf{s},\mathbf{t}}\ge0,\norm{\mathbf{s}}_1=\norm{\mathbf{t}}_1=1\}.
\end{align*}
Note that in this definition we allow the entries of $\mathbf{s}$ and $\mathbf{t}$ to be $0$, thus
$ST_\lambda$ is closed and $ST_\lambda\subset\left\{(\mathbf{s},\mathbf{t})\in\mathbb{R}_{\ge0}^n\times\mathbb{R}_{\ge0}^n\right|\norm{\mathbf{s}}_1=\norm{\mathbf{t}}_1=1\}$ which is a compact set. And also $\forall\ 0\le\lambda\le\lambda'\in \Lambda(P)$, $ST_\lambda\supseteq ST_{\lambda'}\neq\emptyset$, which is directly from Eq.\eqref{eq:derivative}. So we have $\bigcap_{\lambda\in \Lambda(P)}ST_\lambda\neq\emptyset$.

For any $\lambda\in \Lambda(P)$ and $(\mathbf{s}',\mathbf{t}')\in\bigcap_{\lambda\in \Lambda(P)}ST_\lambda$, we always have $\hat{P}_\lambda^{\mathbf{s}',\mathbf{t}'}(x,y)\ge0$ for all $x,y$. Then $\hat{P}_{\sup \Lambda(P)}^{\mathbf{s}',\mathbf{t}'}(x,y)=\lim_{\lambda\rightarrow\sup \Lambda(P)}\hat{P}_{\lambda}^{\mathbf{s}',\mathbf{t}'}(x,y)\ge0$.
That is, for $\lambda_u=\sup \Lambda(P)$, there exist $\mathbf{s},\mathbf{t}\in\mbb{R}_{\ge0}^{n}$ with $\norm{\mathbf{s}}_1=\norm{\mathbf{t}}_1=1$ such that $\hat{P}_{\lambda_u}^{\mathbf{s},\mathbf{t}}(x,y)\ge0$. However, since $\sup \Lambda(P)\notin \Lambda(P)$, there do not exist $\mathbf{s},\mathbf{t}\in\mbb{R}_{>0}^{n}$ with $\norm{\mathbf{s}}_1=\norm{\mathbf{t}}_1=1$ satisfying that $\hat{P}_{\lambda_u}^{\mathbf{s},\mathbf{t}}(x,y)\ge0$.

``$3\Rightarrow2$": Suppose that $\lambda_u$ satisfies Statement 3 and $\hat{P}_{\lambda_u}^{ \mathbf{s}_u, \mathbf{t}_u}(x,y)\ge0$ for all $x,y$, where $ \mathbf{s}_u, \mathbf{t}_u\in\mbb{R}_{\ge0}^{n}$ satisfy $\norm{ \mathbf{s}_u}_1=\norm{ \mathbf{t}_u}_1=1$. We prove that for any $\epsilon>0$, $\lambda_u-\epsilon\in \Lambda(P)$.

For any $x,y$, if $( \mathbf{s}_u)_x=( \mathbf{t}_u)_y=0$, then $\hat{P}_{\lambda-\epsilon}^{ \mathbf{s}_u, \mathbf{t}_u}(x,y)=P(x,y)>0$.
And if $( \mathbf{s}_u)_x=0,( \mathbf{t}_u)_y>0$, we have
\begin{align*}
\frac{\partial \hat{P}_{\lambda}^{ \mathbf{s}_u, \mathbf{t}_u}(x,y)}{\partial \lambda}|_{\lambda=\lambda_u-\epsilon}=-( \mathbf{t}_u)_y\sum_bP(x,b)<0.
\end{align*}
Thus $\hat{P}_{\lambda_u-\epsilon}^{ \mathbf{s}_u, \mathbf{t}_u}(x,y)> \hat{P}_{\lambda_u}^{ \mathbf{s}_u, \mathbf{t}_u}(x,y)\ge0$.
The case that $( \mathbf{s}_u)_x>0$ and $( \mathbf{t}_u)_y=0$ is similiar.

Lastly, if $( \mathbf{s}_u)_x>0,( \mathbf{t}_u)_y>0$, then
\begin{align*}
\frac{\partial \hat{P}_{\lambda}^{ \mathbf{s}_u, \mathbf{t}_u}(x,y)}{\partial \lambda}|_{\lambda=\lambda_u-\epsilon}=\frac{\partial \hat{P}_{\lambda}^{ \mathbf{s}_u, \mathbf{t}_u}(x,y)}{\partial \lambda}|_{\lambda=\lambda_u}-2\epsilon ( \mathbf{s}_u)_x( \mathbf{t}_u)_y\le-2\epsilon ( \mathbf{s}_u)_x( \mathbf{t}_u)_y<0.
\end{align*}
Thus $\hat{P}_{\lambda_u-\epsilon}^{ \mathbf{s}_u, \mathbf{t}_u}(x,y)> \hat{P}_{\lambda_u}^{ \mathbf{s}_u, \mathbf{t}_u}(x,y)\ge0$.

In conclusion, we always have $\hat{P}_{\lambda_u-\epsilon}^{ \mathbf{s}_u, \mathbf{t}_u}(x,y)>0$. Let $w=\underset{x,y}{\min}\hat{P}_{\lambda_u-\epsilon}^{ \mathbf{s}_u, \mathbf{t}_u}(x,y)>0$. Note that
\begin{align*}
\left|\hat{P}_\lambda^{\mathbf{s},\mathbf{t}}(x,y)-\hat{P}_\lambda^{\mathbf{s}',\mathbf{t}'}(x,y)\right|\le3\left| s_x- s'_x\right|+3\left| t_y- t'_y\right|
\end{align*}
for any $\lambda,\mathbf{s},\mathbf{t},\mathbf{s}',\mathbf{t}',x,y$. Then for any $\mathbf{s},\mathbf{t}\in\mathbb{R}_{>0}^n$ that satisfy $\norm{\mathbf{s}- \mathbf{s}_u}_2\le\frac{w}{6}$ and $\norm{\mathbf{t}- \mathbf{t}_u}_2\le\frac{w}{6}$ (such $\mathbf{s},\mathbf{t}$ exist), we have
\begin{align*}
\hat{P}_{\lambda_u-\epsilon}^{\mathbf{s},\mathbf{t}}(x,y)\ge \hat{P}_{\lambda_u-\epsilon}^{ \mathbf{s}_u, \mathbf{t}_u}(x,y)-w\ge0.
\end{align*}
Thus $\lambda_u-\epsilon\in \Lambda(P)$, which means $\Lambda(P)=\left[0,\lambda_u\right)$.

``$2\Rightarrow1$": Let $\lambda_u=\sup \Lambda(P)$. Suppose $\sup_{\lambda\in \Lambda(P)}\mathrm{C}_{\lambda}(P)<d$ for some constant $d$. Then from Eq.\eqref{eq:derivative}, for any $\lambda\in \Lambda(P)$, there exist $\mathbf{s},\mathbf{t}\in\mbb{R}_{>0}^n$ such that $\norm{\mathbf{s}}_1=\norm{\mathbf{t}}_1=1$ and $\hat{P}_\lambda^{\mathbf{s},\mathbf{t}}(x,y)\ge0$ for all $x,y$, where $ s_x\ge\frac{\sum_bP(x,b)}{d(1-\lambda)+\lambda}$ and $ t_y\ge\frac{\sum_aP(a,y)}{d(1-\lambda)+\lambda}$, as shown in Eq.\eqref{eq:proof_of_C_LB}. Let
\begin{align*}
v=\min_{x,y}\left\{\frac{\sum_bP(x,b)}{d(1-\lambda_u)+\lambda_u},\frac{\sum_aP(a,y)}{d(1-\lambda_u)+\lambda_u}\right\}.
\end{align*}
We have $\mathbf{s},\mathbf{t}\in\mbb{R}_{\ge v}^n$. Again let
\begin{align*}
    ST_\lambda=\left\{(\mathbf{s},\mathbf{t})\in\mathbb{R}_{\ge v}^n\times\mathbb{R}_{\ge v}^n\right|\hat{P}_\lambda^{\mathbf{s},\mathbf{t}}\ge0,\norm{\mathbf{s}}_1=\norm{\mathbf{t}}_1=1\}.
\end{align*}
Note that $ST_\lambda$ is closed and $ST_\lambda\subset\left\{(\mathbf{s},\mathbf{t})\in\mathbb{R}_{\ge v}^n\times\mathbb{R}_{\ge v}^n\right|\norm{\mathbf{s}}_1=\norm{\mathbf{t}}_1=1\}$ which is compact. And for any $\lambda\le\lambda'\in \Lambda(P)$, $ST_\lambda\supseteq ST_{\lambda'}\neq\emptyset$, which is directly from Eq.\eqref{eq:derivative}. So we have $\bigcap_{\lambda\in \Lambda(P)}ST_\lambda\neq\emptyset$.

For any $\lambda\in \Lambda(P)$ and $(\mathbf{s}',\mathbf{t}')\in\bigcap_{\lambda\in \Lambda(P)}ST_\lambda$, we always have $\hat{P}_\lambda^{\mathbf{s}',\mathbf{t}'}(x,y)\ge0$ for all $x,y$. Then
\begin{equation}
\hat{P}_{\lambda_u}^{\mathbf{s}',\mathbf{t}'}(x,y)=\hat{P}_{\sup \Lambda(P)}^{\mathbf{s}',\mathbf{t}'}(x,y)=\lim_{\lambda\rightarrow\sup \Lambda(P)}\hat{P}_{\lambda}^{\mathbf{s}',\mathbf{t}'}(x,y)\ge0.
\end{equation}
Thus $\sup \Lambda(P)\in \Lambda(P)$, which contradicts $\Lambda(P)$ is right open. Thus we have $\sup_{\lambda\in \Lambda(P)}\mathrm{C}_{\lambda}(P)=\infty$.
\end{proof}

\end{document}